\documentclass[journal,onecolumn,11pt]{IEEEtran}
\usepackage{amsmath,amssymb,amsthm,mathtools,bm}
\usepackage{booktabs,array,multirow}
\usepackage{graphicx}
\usepackage{cite}
\usepackage[hidelinks]{hyperref}
\usepackage{url}

\newtheorem{theorem}{Theorem}
\newtheorem{proposition}{Proposition}
\newtheorem{lemma}{Lemma}
\newtheorem{corollary}{Corollary}
\theoremstyle{definition}

\theoremstyle{remark}
\newtheorem{remark}{Remark}

\newcommand{\E}{\mathbb E}
\newcommand{\PP}{\mathbb P}
\newcommand{\R}{\mathbb R}
\newcommand{\C}{\mathbb C}
\newcommand{\Z}{\mathbb Z}
\newcommand{\Rea}{\operatorname{Re}}

\newcommand{\tr}{\operatorname{tr}}
\newcommand{\dist}{\operatorname{dist}}
\newcommand{\dens}{\operatorname{dens}}
\newcommand{\covol}{\operatorname{covol}}
\newcommand{\conv}{\operatorname{conv}}
\newcommand{\diag}{\operatorname{diag}}
\newcommand{\argmin}{\operatorname*{arg\,min}}
\newcommand{\op}{\mathrm{op}}
\newcommand{\HS}{\mathrm{HS}}
\newcommand{\CN}{\mathcal{CN}}
\newcommand{\V}{\mathcal V}
\newcommand{\D}{\mathcal D}
\newcommand{\G}{\mathcal G}
\newcommand{\A}{\mathcal A}

\newcommand{\norm}[1]{\lVert#1\rVert}
\newcommand{\ip}[2]{\langle#1,#2\rangle}
\newcommand{\dd}{\,\mathrm d}

\begin{document}
\title{Deep Holes and Aperture: Sign Laws for the Empirical Correlation of Lattice-Sampled Random Wave Fields}
\author{\IEEEauthorblockN{Tianhao Guo\IEEEauthorrefmark{1},
Hao Luo\IEEEauthorrefmark{2},
Xianzhong Li\IEEEauthorrefmark{3}, and
Yuzhi Yang\IEEEauthorrefmark{4}}\\
\IEEEauthorblockA{\IEEEauthorrefmark{1}College of Physics and Electronic Engineering, Shanxi University, Taiyuan 030006, China. Email: tianhao\_guo@sxu.edu.cn}\\
\IEEEauthorblockA{\IEEEauthorrefmark{2}College of Physics and Electronic Engineering, Shanxi University, Taiyuan 030006, China}\\
\IEEEauthorblockA{\IEEEauthorrefmark{3}College of Physics and Electronic Engineering, Shanxi University, Taiyuan 030006, China}\\
\IEEEauthorblockA{\IEEEauthorrefmark{4}College of Information Science, Japan Advanced Institute of Science and Technology, Nomi 9231292, Japan. Email: yangyzis@jaist.ac.jp}\\
\IEEEauthorblockA{Corresponding author: Yuzhi Yang}}
\maketitle

\begin{abstract}
A Gaussian random wave field is sampled on a lattice within a bounded aperture. After displacement, re-indexing leaves a residual in the Voronoi cell. We study whether the empirical lag correlation has positive real part throughout the cell. For each realization in a single-output reciprocal link with noiseless over-the-air return, this is equivalent to the negative returned gradient being a descent direction for all residuals and differentiable losses with nonzero error signals. Return noise is treated exactly. For isotropic fields, ball-shaped apertures, unshifted lattice points, and dilations approaching the critical covering radius, every minimizer is a deep hole except with exponentially small probability. A Gaussian vector on the deep holes determines the limiting failure probability, with explicit covariance and lattice-independent limiting normalization. Square, hexagonal, simple cubic, and body-centered cubic lattices have two, three, four, and twelve sign classes. At fixed target failure probabilities below the critical limiting value, this covariance determines the exact leading aperture constant. When the margin and target failure probability vanish jointly, that constant is common to lattice shapes within each dimension. A cubic correction yields vanishing relative tail error in a moderate-deviation regime. A non-asymptotic two-sided aperture law has explicit constants. With a Rayleigh-faded specular path and positive diffuse margin, a failure floor occurs at every positive path power exactly when the cosine of its phase change is negative somewhere in the cell. The floor vanishes for all directions exactly when the covering radius is at most a quarter wavelength.
\end{abstract}


\section{Introduction}

Let a monochromatic random wave field be sampled at the points of a lattice $\Lambda$ inside a bounded aperture. Translating the sample positions by a lattice vector moves each sample onto the position of another sample, so a displacement $\tau$ of the whole array can be undone by re-indexing, up to a residual $e=\tau-\ell$ in the Voronoi cell $\V$ and the samples that leave the aperture. Whether the field seen after the re-indexing still agrees with the field before it is measured by the empirical lag correlation
\begin{equation}\label{eq:intro-S}
 S_L(e)=\frac1{|A^\ell_L|}\sum_{a\in A^\ell_L}h(a)\overline{h(a+e)},
\end{equation}
where $A^\ell_L$ is the set of samples whose translated partner exists. We study the event $\{\Rea S_L(e)>0\ \text{for all }e\in\V\}$.

The event is motivated by the following application. Reciprocity lets a transmitter obtain the gradient of a receiver-side loss with respect to its transmit vector from one analog return transmission~\cite{ota}. If an array of densely placed ports, as in port-selection and fluid-antenna architectures~\cite{wong}, moves between the forward and the return slot and re-indexes its ports by $\ell$, then for a single-output link
\begin{equation}\label{eq:intro-identity}
 \Rea\ip{g}{\widehat g}=|A^\ell_L|\,|\delta|^2\,\Rea S_L(e)
\end{equation}
for every loss, every error signal $\delta$, and every realization of the field, when the return transmission is noiseless (Lemma~\ref{lem:identity}, with return noise treated in Proposition~\ref{prop:noise}). For $\delta\ne0$ the sign of the inner product is the sign of $\Rea S_L(e)$. With receiver noise on the return, the conditional probability that the noisy inner product is nonpositive is $\Phi(-\sqrt{2M\rho/Q_L}\,\Rea S_L(e))$, where $M=|A^\ell_L|$, $\rho$ is the return signal-to-noise ratio and $Q_L$ the empirical field power. On $\{\Rea S_L(e)\ge\gamma/2,\ Q_L\le2\}$ it is at most $\exp(-\gamma^2M\rho/8)$. Here $g$ is the true gradient and $\widehat g$ the returned one. The lattice is fixed before the motion. For a general nonzero residual, knowing $e$ alone does not determine the realized sign of the empirical correlation. We study a guarantee that holds simultaneously for all residuals in the cell. For every nonzero error signal, the update direction $-\widehat g$ is a descent direction for every displacement with the same $\ell$ exactly on the event above.

The mean of $\Rea S_L(e)$ is $\Rea\phi(e)$, where $\phi$ is the spatial correlation. Its minimum over the cell is the margin $\gamma$. For isotropic scattering and a covering radius $R\le r_0$, the first zero of $\phi$, the margin is the value of $\phi$ at distance $R$. A positive mean for every residual means $R<r_0$, and the sparsest such lattices are the optimal coverings~\cite{bambah,barnes56}. A realization of $\Rea S_L$ fluctuates around its mean on the scale $\sigma_L\asymp L^{-(n-1)/2}$. With $M$ samples in dimension $n\in\{2,3\}$, this is not $M^{-1/2}\asymp L^{-n/2}$: the spectrum of a monochromatic field lives on a sphere, and an aperture of linear size $L$ resolves about $L^{n-1}$ spectral directions. The question is how the margin, the fluctuation, and the geometry of the cell combine to decide the sign of $\min_\V\Rea S_L$.

\subsection{Contributions}

\begin{enumerate}
\item \emph{Deep-hole limit law and exact critical aperture} (Theorems~\ref{thm:localization}--\ref{thm:joint}). This is the main result. For an isotropic field with continuous spectrum, the window $K=\overline{B_1}$ (aperture of radius $L$), the unshifted sample set $\ell=0$, and dilations of a fixed lattice shape approaching the critical one, the minimizer of $\Rea S_L$ over $\V$ is a deep hole except on an event of probability $Ce^{-cL^{n-1}}$. The failure probability converges to $\PP(\min_{v\in\D}Z_v\le-t)$, with $t=\lim\gamma/\sigma_L$ and $Z$ a centered Gaussian vector on the deep holes with covariance $\frac12[\phi(v-w)+\phi(v+w)]$. The normalization $L^{n-1}\sigma_L^2\to\kappa_n$ is the same for every lattice. Lattices whose reciprocal vectors alias or are tangent to the sphere of radius $2k$ are included. Deep holes pair up under $v\mapsto-v$. The number of classes is $2$, $3$, $4$, $12$ for the square, hexagonal, simple cubic, and body-centered cubic (BCC) lattices. At the critical radius, distinct classes are negatively correlated for the first three. For BCC the correlations have both signs. In the critical scaling limit at fixed $\eta<p_0$, the aperture satisfies $L_\eta^{n-1}\sim\kappa_nq_\eta^2\gamma^{-2}$, where $q_\eta$ is a quantile of $\min Z$. The lattice therefore enters the constant through its deep-hole configuration, beyond the covering radius. When $\gamma\to0$ and $\eta\to0$ jointly, at arbitrary relative rates, $L_\eta^{n-1}\sim2\kappa_nG_*\gamma^{-2}\log(1/\eta)$, where $G_*=\frac12(1+\psi(2r_0))$, with $\psi$ the radial profile of $\phi$, is the variance at the deep holes of the critical lattice. Since these deep holes lie at distance $r_0$, $G_*$ is the same for every lattice: at high reliability the lattice enters the leading constant only through the margin $\gamma$. The number of sign classes enters a lower-order term of the limiting Gaussian quantile. A corresponding expansion of the actual aperture remains open. The Gaussian proxy has vanishing relative error for $\gamma/\sigma_L=o(L^{(n-1)/6})$, and a cubic correction extends this to $o(L^{(n-1)/4}/\sqrt{\log L})$ (Theorem~\ref{thm:joint}). For a fixed window and a fixed continuous density the covariance limit carries a chord-length weight.
\item \emph{Two-sided aperture law} (Theorem~\ref{thm:aperture}). For continuous spectra with $\gamma>0$, within the compact family~\eqref{eq:family}, every $L\ge L_0$, and every $\eta\in(0,1)$, $L^{n-1}\ge C_3\gamma^{-2}[1+\log(1/\eta)]$ gives failure probability at most $\eta$, and $L\le c_4\gamma^{-2/(n-1)}$ gives at least $1/144$. The constants are explicit and uniform over a compact family of lattices, windows, shifts, and spectral densities. For each fixed $\eta<1/144$ this gives $L^{n-1}\asymp\gamma^{-2}$. Over the admissible family the logarithm cannot be replaced uniformly by any $o(\log(1/\eta))$, and at $\gamma=0$ the failure probability is at least $1/36$ for every aperture. Theorem~\ref{thm:joint} identifies the exact constant in this order near the critical lattice.
\item \emph{Specular components} (Theorems~\ref{thm:specular}--\ref{thm:known}). With a Rayleigh-faded specular path and a positive continuous margin, the failure probability has a positive limit if and only if $\cos(ku\cdot e)<0$ somewhere in $\V$, for every positive path power. A fixed-amplitude path has instead an exact power threshold. No direction produces a floor if and only if $R\le\lambda/4$. This needs port spacings below $\lambda/2$, where mutual coupling, which the correlation model ignores, may matter (Section~\ref{sec:discussion}). For a known direction in the plane, the optimal lattice is explicit, with $2.974$ points per $\lambda^2$, against $2.627$ for diffuse scattering and $6.158$ for an unknown direction.
\end{enumerate}

\subsection{Related Work}

\emph{Lattices.} Covering and quantization optimality in two and three dimensions are classical~\cite{bambah,barnes56,conway,barnes}. Petersen and Middleton identified optimal sampling lattices for band-limited fields~\cite{pm62}. K\"unsch, Agrell and Hamprecht~\cite{kunsch} and Agrell and Cs\'ebfalvi~\cite{agrell18} identified optimal lattices for interpolation of random fields. These criteria select lattices through packing or covering properties of the lattice or of its dual, depending on the spectrum and on the sampling rate. Here the lattice acts through the residual of a translation. Since the guarantee must hold for every residual, the deep holes of the cell, rather than its normalized second moment~\cite{conway}, govern the extreme value of a random field. The planar problem with a known specular direction is a lattice covering by a convex body, settled by F\'ary's theorem~\cite{fary}.

\emph{Random waves.} Extremes and excursions of Gaussian fields are treated in~\cite{adlertaylor,azais}, Gaussian fields on the sphere in~\cite{marinucci}, and normal approximation for Wiener chaos in~\cite{nourdinpeccati}. Nodal statistics of random waves, for instance of arithmetic random waves~\cite{kkw}, are a central example. Two differences matter here. For nodal length the leading-order variance cancels, whereas $\Rea S_L(e)$ is itself a quadratic form in the field, and its variance is of order $L^{-(n-1)}$ for every residual (Proposition~\ref{prop:scale}). For arithmetic random waves the field is built from lattice points on a circle, whereas here the field has a continuous spectrum on $\mathbb R^n$ and the lattice only samples it. Quadratic functionals of Berry's monochromatic random wave~\cite{berry} have Gaussian limits whose variance involves the geometry of the integration domain. In the continuum, squared line-integral structures of this kind appear in~\cite{gps}, and spectral universality of excursion volumes is studied in~\cite{beliaev}. Asymptotic covariances for functionals of stationary fields over several domains are computed in~\cite{maini}, and a spectral central limit theorem for additive functionals of isotropic stationary fields is proved in~\cite{maininourdin}. Our normalization constant and covariance limit are the lattice-sampled, hard-window counterparts. The new point is that they hold uniformly through reciprocal aliases and tangencies, and along a family of lattices approaching the critical one. That uniformity is what the inversion to an exact aperture requires.

\emph{Extremes.} That a minimum over a compact set localizes at a sharp minimizer of the mean under small Lipschitz perturbations is classical in optimization~\cite{ferris}. That the law of a minimum is governed by the minimizing set of the mean is the directional delta method for infimum functionals~\cite{carcamo}. We use both mechanisms. The work lies in verifying their hypotheses for this family of chaos processes with exponentially small exceptional probability, and in identifying the limit covariance on the deep holes.

\emph{Spatial averaging in wireless channels.} Channel hardening is the concentration of the normalized channel energy around its mean. Its counterpart for inner products between the channels of different users, favorable propagation, is discussed in~\cite{ngo14}. For finitely many nonaliasing rays, random ray amplitudes can prevent hardening, whereas fixed amplitudes with random phases do not~\cite{liray}. For a moving terminal, channel aging describes the temporal mismatch between the channel used for processing and the channel at transmission~\cite{truong}. The planar correlation model is Clarke's~\cite{clarke}. Theorem~\ref{thm:specular} is the sign version of non-hardening: it identifies which specular directions produce a floor for $\min_\V\Rea S_L$, as opposed to a non-vanishing variance.

\emph{Tools.} The tail estimates are Hanson--Wright type bounds for Gaussian quadratic forms~\cite{rv13}. The normal approximation is proved directly by Gaussian interpolation (Appendix~\ref{app:deepholes}). The problem-specific inputs are an operator-norm bound for the sampling operator on the sphere that is uniform through tangencies, a lower bound for the variance of the real part, and the linear growth of the mean away from the deep holes.

\subsection{Organization}

Section~\ref{sec:model} sets up the model and proves~\eqref{eq:intro-identity}. Section~\ref{sec:geometry} records the geometry of the mean. Section~\ref{sec:aperture} proves the two-sided law, Section~\ref{sec:deepholes} the deep-hole law, and Section~\ref{sec:specular} treats specular components. Section~\ref{sec:numerics} compares the deep-hole law with independent simulations, and Section~\ref{sec:discussion} discusses scope and open problems. Proofs are in the appendices.

\emph{Notation.} Lengths are measured in wavelengths: $\lambda=1$ and $k=2\pi$. The dimension is $n\in\{2,3\}$ and $d=n-1$. $\sigma$ is the normalized surface measure on $S^{n-1}$. For an operator $T$, $\norm{T}_\op$ and $\norm{T}_\HS$ are the operator and Hilbert--Schmidt norms. $\Phi$ is the standard normal distribution function. $B_r$ denotes the closed Euclidean ball of radius $r$ centered at the origin.

\section{Model and the Exact Reduction}\label{sec:model}

\subsection{Field}
Let $h:\R^n\to\C$ be a zero-mean, stationary, circularly symmetric complex Gaussian field with $\E|h(x)|^2=1$ and covariance
\begin{equation}\label{eq:phi}
 \phi(r)=\E\big[h(x+r)\overline{h(x)}\big]=\int_{S^{n-1}}e^{\mathrm iku\cdot r}\,\mu(\dd u),
\end{equation}
where $\mu$, the angular spectrum, is a probability measure on the unit sphere. The field is band-limited and has a $C^\infty$ version, which we use. We write $m(e)=\Rea\phi(e)$. It is even in $e$. The \emph{continuous-spectrum} case is $\mu=f\sigma$ with $0<c_f\le f\le C_f<\infty$ almost everywhere. The \emph{isotropic} case $f\equiv1$ gives Clarke's planar correlation~\cite{clarke} and its spherical analogue,
\begin{equation}\label{eq:radial}
 \phi(r)=J_0(2\pi|r|)\ (n=2),\qquad \phi(r)=\frac{\sin(2\pi|r|)}{2\pi|r|}\ (n=3),
\end{equation}
where $J_0$ is the Bessel function of the first kind of order zero and $j_{0,1}$ is its smallest positive zero. The first zeros are $r_0=j_{0,1}/(2\pi)=0.38274\ldots$ and $r_0=1/2$. Mixed spectra with atoms are treated in Section~\ref{sec:specular}.

The field is the same in the forward and the return slot: the scatterers are frozen over one round trip and the link is reciprocal. The probability of the failure event studied below is taken over $h$.

\subsection{Lattice array, translation, and surviving ports}
Let $\Lambda=B\Z^n$ be a full-rank lattice with covolume $v=|\det B|$, closed Voronoi cell
\[
 \V=\{e\in\R^n:\ |e|\le|e-\ell|\ \ \forall\ell\in\Lambda\},
\]
covering radius $R=\max_{e\in\V}|e|$, and deep-hole set $\D=\{e\in\V:|e|=R\}$, a finite set of vertices of $\V$. Let $K$ be a convex body with $K=-K$ and $0\in\operatorname{int}K$. The ports occupy $\Lambda\cap LK$.

Between the forward and the return slot the array is displaced by $\tau$. Choose a nearest lattice point $\ell\in\Lambda$ to $\tau$, with a fixed tie-breaking rule on cell boundaries, and write $\tau=\ell+e$ with $e\in\V$. In the return slot the device lets the port with index $a-\ell$ play the role of forward port $a$. That port now sits at $a-\ell+\tau=a+e$. The translation is available for forward port $a$ exactly when $a-\ell\in LK$, so the surviving forward ports are
\begin{equation}\label{eq:survivors}
 A_L^\ell=\Lambda\cap LK\cap(LK+\ell),\qquad M=|A_L^\ell|.
\end{equation}
We assume $M>0$ throughout.
The translation $\ell$ depends only on the displacement and the lattice (Fig.~\ref{fig:geometry}(a)).

\begin{figure*}[t]
\centering
\includegraphics[width=\textwidth]{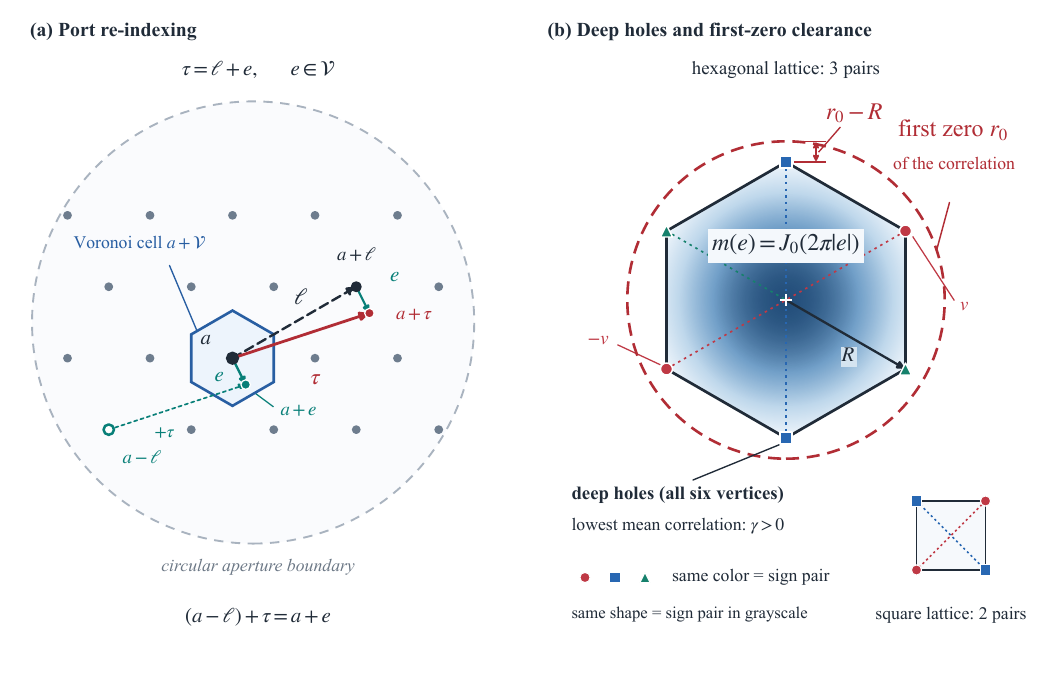}
\caption{(a) A displacement $\tau$ of the port array is undone by re-indexing with the nearest lattice vector $\ell$. The port with index $a-\ell$ lands at $a+e$, inside the Voronoi cell of $a$. (b) The mean correlation $m(e)$ on the cell of the hexagonal lattice is lowest at the six deep holes, where it equals the margin $\gamma>0$ when the covering radius $R$ is below the first zero $r_0$. Opposite deep holes form a sign pair, giving three pairs for the hexagonal lattice and two for the square lattice (inset).}
\label{fig:geometry}
\end{figure*}

\subsection{Single-output gradient return}
In the forward slot the device transmits $x\in\C^{\Lambda\cap LK}$ and a single-antenna receiver observes $z=\sum_a h(a)x_a+w$, where $w$ denotes forward receiver noise. Let $\mathcal L$ be a real-valued loss, evaluated at the receiver and differentiable at $z$. The Wirtinger gradient with respect to $\bar x$ is
\begin{equation}\label{eq:truegrad}
 g_a=\delta\,\overline{h(a)},\qquad \delta=\partial\mathcal L/\partial\bar z.
\end{equation}
In the return slot the receiver transmits $\bar\delta$. Port $a-\ell$, now at $a+e$, receives $h(a+e)\bar\delta$ through the reciprocal channel. After conjugation the device forms
\begin{equation}\label{eq:returngrad}
 \widehat g_a=\delta\,\overline{h(a+e)}\quad(a\in A_L^\ell),\qquad \widehat g_a=0\quad\text{otherwise}.
\end{equation}
This is the over-the-air backward pass of~\cite{ota} with port re-indexing and without return noise.

\begin{lemma}[Exact reduction]\label{lem:identity}
For every realization of the field, every loss, and every $\delta\in\C$,
\[
 \Rea\ip g{\widehat g}=M|\delta|^2\,\Rea S_L(e),\qquad S_L(e):=\frac1M\sum_{a\in A_L^\ell}h(a)\overline{h(a+e)}.
\]
In particular, if $\delta\ne0$, then $-\widehat g$ is a first-order descent direction if and only if $\Rea S_L(e)>0$.
\end{lemma}

\begin{proof}
$\ip g{\widehat g}=\sum_a\overline{g_a}\,\widehat g_a=|\delta|^2\sum_{a\in A_L^\ell}h(a)\overline{h(a+e)}$. Ports outside $A_L^\ell$ contribute nothing. Take real parts.
\end{proof}

The error signal $\delta$ may depend on the forward channel, the data, and the loss, and enters only through the factor $|\delta|^2$. Positivity of $\Rea S_L(e)$ is therefore a statement about all losses simultaneously, and port loss enters only through the sample set $A_L^\ell$.

With return noise, the identity of Lemma~\ref{lem:identity} holds for the conditional mean, and the noisy inner product is nonpositive with the following conditional probability. When $\Rea S_L(e)>0$, this is the probability that noise reverses the sign.

\begin{proposition}[Return noise]\label{prop:noise}
Let $\delta\ne0$ and $\sigma_w^2>0$. Suppose port $a-\ell$ receives $h(a+e)\bar\delta+w_a$, with $w_a$ independent and identically distributed (i.i.d.) circularly symmetric complex Gaussian variables, written $w_a\sim\CN(0,\sigma_w^2)$ with $\E|w_a|^2=\sigma_w^2$, independent of the forward information $(h,\delta)$, and the device forms $\widehat g_a$ as the conjugate of what it receives. Let $\rho=|\delta|^2/\sigma_w^2$ and $Q_L=M^{-1}\sum_{a\in A_L^\ell}|h(a)|^2$. Then
\[
 \PP\big(\Rea\ip g{\widehat g}\le0\,\big|\,h,\delta\big)=\Phi\Big(-\sqrt{2M\rho/Q_L}\;\Rea S_L(e)\Big).
\]
In particular, if $\gamma>0$, then on the event $\{\Rea S_L(e)\ge\gamma/2,\ Q_L\le2\}$ this probability is at most $\exp(-\gamma^2M\rho/8)$.
\end{proposition}

\begin{proof}
Now $\widehat g_a=\delta\,\overline{h(a+e)}+\overline{w_a}$, so
\[
 \ip g{\widehat g}=M|\delta|^2S_L(e)+\bar\delta\sum_ah(a)\overline{w_a}.
\] Given $(h,\delta)$, the sum is $\CN(0,\sigma_w^2MQ_L)$, and the real part of $\bar\delta$ times it is centered normal with variance $|\delta|^2\sigma_w^2MQ_L/2$. The bound uses $\Phi(-x)\le e^{-x^2/2}$ for $x\ge0$.
\end{proof}

For continuous spectra in the family of Section~\ref{sec:aperture}, $Q_L=S_L(0)$ has mean $1$, and Lemmas~\ref{lem:op} and~\ref{lem:tail} give $\PP(Q_L>2)\le2e^{-cL^{n-1}}$. With Rayleigh-faded specular paths, random path energy can leave a positive limiting value of $\PP(Q_L>2)$. Since $M\asymp L^n$ and Theorem~\ref{thm:aperture} needs $L^{n-1}\gtrsim\gamma^{-2}$, the exponent $\gamma^2M\rho$ is at least of order $L\rho$, and of exactly that order at the critical scaling $\gamma^2L^{n-1}\asymp1$. To keep the noise contribution below $\eta_{\rm noise}$, it suffices that $\gamma^2M\rho\ge8\log(1/\eta_{\rm noise})$, in addition to controlling the events $\{\min_\V\Rea S_L<\gamma/2\}$ and $\{Q_L>2\}$.

\subsection{The event of interest}
Since the residual is not known when the lattice is designed, we require positivity for every residual. The \emph{margin} and the \emph{failure event} are
\begin{equation}\label{eq:failure}
 \gamma=\min_{e\in\V}m(e),\qquad F_L=\Big\{\min_{e\in\V}\Rea S_L(e)\le0\Big\}.
\end{equation}
$F_L$ is defined for a fixed shift $\ell$. For finitely many shifts, the union bound adds a factor equal to their number.

The fluctuation scale is
\begin{equation}\label{eq:sigma}
 \sigma_L^2=\E|S_L(e)-\E S_L(e)|^2=\frac1{M^2}\sum_{a,b\in A_L^\ell}|\phi(a-b)|^2,
\end{equation}
independent of $e$ (Lemma~\ref{lem:wick}). Note that $\E S_L(e)=\phi(-e)$, so $\E\Rea S_L(e)=m(e)$.

\section{Geometry of the Mean}\label{sec:geometry}

This section records the behavior of the mean $m=\Rea\phi$ on the cell.

\begin{proposition}[Margin and deep holes]\label{prop:margin}
In the isotropic case, $\phi(r)=\psi(|r|)$ with $\psi$ strictly decreasing on $[0,r_0]$, and
\[
 \gamma=\min_{0\le r\le R}\psi(r).
\]
If $R\le r_0$, then $\gamma=\psi(R)$, and the minimizers of $m$ on $\V$ are exactly the deep holes. Consequently $\gamma>0$ if and only if $R<r_0$.
\end{proposition}

\begin{proof}
The radial function satisfies $(r^{n-1}\psi')'=-k^2r^{n-1}\psi$ with $\psi'(0)=0$, so $r^{n-1}\psi'(r)=-k^2\int_0^rt^{n-1}\psi(t)\dd t<0$ for $0<r\le r_0$. The residual norms attained in $\V$ fill $[0,R]$ by convexity, and $|e|=R$ exactly on $\D$.
\end{proof}

Without isotropy the comparison with $r_0$ fails. For example, $f(\theta)=1+0.9\cos2\theta$ in the plane and $\Lambda=\diag(2/3,1/5)\Z^2$ give $R=\sqrt{109}/30<r_0$ but $m(1/3,0)<-0.096$. In general we work with $\gamma$ directly.

\begin{proposition}[Port density at a given margin]\label{prop:density}
Let $\gamma\in(0,1)$ and let $r_\gamma\in(0,r_0)$ solve $\psi(r_\gamma)=\gamma$. In the isotropic case, the least density of a lattice of a given shape with margin at least $\gamma$ is $(\vartheta/r_\gamma)^n$, where $\vartheta=R/v^{1/n}$ is the normalized covering radius. Among all lattices it is attained by the hexagonal lattice for $n=2$ and by the BCC lattice for $n=3$. The square and simple cubic lattices require, respectively, $3\sqrt3/4\approx1.299$ and $12\sqrt3/(5\sqrt5)\approx1.859$ times as many points.
\end{proposition}

\begin{proof}
$R=\vartheta\,\dens(\Lambda)^{-1/n}$, and by Proposition~\ref{prop:margin} the margin constraint is $R\le r_\gamma$. Minimizing $\vartheta$ is the lattice covering problem. In the plane it is solved by the hexagonal lattice~\cite[Thm.~2]{barnes56}. In space, Barnes~\cite[Thm.~1]{barnes56}, following Bambah~\cite{bambah}, shows that the extremal quadratic form has diagonal entries $3$ and off-diagonal entries $-1$. It is the Gram matrix of the BCC basis $(-1,1,1)$, $(1,-1,1)$, $(1,1,-1)$. For BCC with conventional cube side $a$, $v=a^3/2$ and $R=\sqrt5a/4$, so the least density at covering radius $R$ is $5\sqrt5/(32R^3)$. For the hexagonal lattice it is $2/(3\sqrt3R^2)$. The normalized radii are $\vartheta^2=2/(3\sqrt3)$ (hexagonal), $1/2$ (square), $\vartheta=\sqrt5\,2^{1/3}/4$ (BCC), and $\sqrt3/2$ (simple cubic).
\end{proof}

Table~\ref{tab:densities} lists the densities at the two radii that appear below: the first zero $r_0$, relevant for diffuse scattering, and $1/4$, relevant when a specular path of unknown direction may be present (Corollary~\ref{cor:quarter}).

\begin{table}[t]
\centering
\caption{Points per $\lambda^n$ at covering radius $R$, and nearest-neighbor spacing in wavelengths. The densities at $R=r_0$ are infima under the constraint $\gamma>0$.}
\label{tab:densities}
\begin{tabular}{lcccc}
\toprule
 & \multicolumn{2}{c}{$n=2$} & \multicolumn{2}{c}{$n=3$}\\
\cmidrule(lr){2-3}\cmidrule(lr){4-5}
 & hexagonal & square & BCC & simple cubic\\
\midrule
$R=r_0$ (diffuse) & 2.627 & 3.413 & 2.795 & 5.196\\
$R=1/4$ (any specular direction) & 6.158 & 8.000 & 22.36 & 41.57\\
\midrule
nearest-neighbor spacing at $R=r_0$ & 0.663 & 0.541 & 0.775 & 0.577\\
nearest-neighbor spacing at $R=1/4$ & 0.433 & 0.354 & 0.387 & 0.289\\
\bottomrule
\end{tabular}
\end{table}

\section{The Two-Sided Aperture Law}\label{sec:aperture}

Fix a compact family of configurations:
\begin{equation}\label{eq:family}
\begin{gathered}
 \text{singular values of }B\text{ in }[b_{\min},b_{\max}],\qquad B_{r_{\rm in}}\subset K\subset B_{r_{\rm out}},\qquad |\ell|\le T,\\
 \mu=f\sigma,\quad c_f\le f\le C_f,\qquad L\ge L_0:=\max\{1,2(T+R_*)/r_{\rm in}\},\quad R_*=\sqrt n\,b_{\max}/2 .
\end{gathered}
\end{equation}
Constants called \emph{family constants} depend only on $(n,b_{\max},r_{\rm in},r_{\rm out},T,c_f,C_f)$. The aperture enters through the sampling operator on $L^2(\mu)$,
\[
 A_L=\frac1M\sum_{a\in A_L^\ell}|v_a\rangle\langle v_a|,\qquad v_a(u)=e^{\mathrm ika\cdot u},
\]
a positive trace-one operator with kernel $A_L(u,u')=M^{-1}\sum_ae^{\mathrm ika\cdot(u-u')}$. Lemma~\ref{lem:wick} shows $\sigma_L^2=\norm{A_L}_\HS^2$.

\begin{proposition}[Fluctuation scale]\label{prop:scale}
There are family constants $C_*,c_*>0$ such that for $L\ge L_0$ and $e\in\V$,
\[
 \norm{A_L}_\op\le C_*L^{-d},\qquad c_*L^{-d}\le\operatorname{Var}\Rea S_L(e)\le\sigma_L^2\le C_*L^{-d}.
\]
The bounds hold for every lattice in the family, including those with reciprocal vectors of length $2k$, and use only $c_f\le f\le C_f$.
\end{proposition}

The order is $L^{-d}$ rather than $M^{-1}\asymp L^{-n}$. The operator bound controls the tails of $\Rea S_L$ through Lemma~\ref{lem:tail}.

\begin{theorem}[Two-sided aperture law]\label{thm:aperture}
In the family~\eqref{eq:family}, there are explicit family constants $C_3,c_4>0$ (Appendix~\ref{app:aperture}) such that the following hold for every $L\ge L_0$.
\begin{enumerate}
\item If $\gamma>0$, $0<\eta<1$, and $L^{d}\ge C_3\gamma^{-2}[1+\log(1/\eta)]$, then $\PP(F_L)\le\eta$.
\item If $\gamma>0$ and $L\le c_4\gamma^{-2/d}$, then $\PP(F_L)\ge1/144$.
\item If $\gamma=0$, then $\PP(F_L)\ge1/36$. If $\gamma<0$, then $\PP(F_L)\ge1-C_*L^{-d}/\gamma^2$.
\end{enumerate}
\end{theorem}

Combined with Lemma~\ref{lem:identity}, Part 1 gives a direction guarantee for a single realization of the field.

\begin{corollary}[Per-realization direction guarantee]\label{cor:direction}
In the family~\eqref{eq:family}, let $\gamma>0$, $0<\eta<1$, $L\ge L_0$, and $L^{d}\ge C_3\gamma^{-2}[1+\log(1/\eta)]$. With probability at least $1-\eta$ over the field, simultaneously for every residual $e\in\V$, every loss, and every error signal $\delta$,
\[
 \Rea S_L(e)\ge\tfrac38\gamma\qquad\text{and}\qquad \Rea\ip g{\widehat g}\ge\tfrac38\gamma\,M|\delta|^2 .
\]
In particular, for $\delta\ne0$ the update direction $-\widehat g$ is a first-order descent direction for every residual in the cell.
\end{corollary}

\begin{proof}
On the complement of the exceptional event in the proof of Part 1 (Appendix~\ref{app:aperture}), $\sup_\V|\Rea S_L-m|\le5\gamma/8$, hence $\Rea S_L(e)\ge m(e)-5\gamma/8\ge3\gamma/8$ on $\V$. Lemma~\ref{lem:identity} converts this into the inner-product bound.
\end{proof}

For each fixed $\eta<1/144$, Parts 1 and 2 give $L_\eta^{\,d}\asymp\gamma^{-2}$ for the critical aperture. Part 3 shows that a positive margin is necessary for reliable operation at any aperture: when the mean vanishes at some residual, enlarging the array cannot make the failure probability tend to zero.

The dependence on $\eta$ is sharp in the following sense.

\begin{proposition}[Sharpness of the logarithm]\label{prop:log}
There is an admissible family, consisting of cubic lattices $a\Z^3$ with $a\in[1/(2\sqrt3),1/\sqrt3]$, a cubic window, and $f\equiv1$, with the following property. No bound of the form $L^d\ge C\gamma^{-2}[1+g(\log(1/\eta))]$ with $g(t)=o(t)$ implies $\PP(F_L)\le\eta$ for all members of the family and all $\eta$.
\end{proposition}

The upper bound in Theorem~\ref{thm:aperture} combines the chaos tail of Lemma~\ref{lem:tail} with chaining over a net of $\V$ that does not depend on $\gamma$, which is why no $\log(1/\gamma)$ factor appears. The lower bound applies a one-sided Paley--Zygmund inequality to $\Rea S_L(e^*)-\gamma$ at a minimizer $e^*$ of $m$.

\section{Localization at Deep Holes and the Limit Law}\label{sec:deepholes}

Theorem~\ref{thm:aperture} determines the order of the critical aperture, and within a fixed admissible family this order depends on the lattice through $\gamma$. This section identifies the location and the limiting law of $\min_\V\Rea S_L$ along lattices approaching the critical one, where $\gamma\to0$, and from them the constant.

\emph{Setting.} In this section the spectrum is isotropic, $K$ is the closed unit ball, and $\ell=0$, so $A_L=\Lambda\cap\overline{B_L}$. Fix a lattice $\Lambda_0$ with covering radius $R_0$ and deep holes $\D_0$, and dilate it: $\Lambda_s=s\Lambda_0$, $\V_s=s\V_0$, $\D_s=s\D_0$. With $s_*=r_0/R_0$ and a fixed $s_-\in(0,s_*]$, let $s\in I=[s_-,s_*]$. By Proposition~\ref{prop:margin}, $\gamma_s=\psi(sR_0)$, attained exactly on $\D_s$, and $\gamma_s\downarrow0$ as $s\uparrow s_*$. Put
\[
 Y_L(e)=\Rea S_L(e)-m(e),\qquad G(v,w)=\tfrac12\big[\phi(v-w)+\phi(v+w)\big].
\]
Unless stated otherwise, constants depend only on $n$, $\Lambda_0$, and $s_-$. Estimates that are uniform over a compact set of lags may also depend on that set, and the $o(1)$ in Theorem~\ref{thm:sharp} may depend on the fixed $\eta$.

\begin{theorem}[Localization and limit law]\label{thm:localization}
\leavevmode
\begin{enumerate}
\item \emph{Lattice-independent limiting normalization.} Uniformly in $s\in I$,
\[
 L^{d}\sigma_L^2\to\kappa_n,\qquad \kappa_2=\frac{8}{3\pi^3},\qquad\kappa_3=\frac{9}{32\pi^2},
\]
and $\operatorname{Cov}(Y_L(e),Y_L(w))/\sigma_L^2\to G(e,w)$ uniformly for $e,w$ in compact sets.
\item \emph{Exact localization.} There are $c,C,L_1$ such that for all $s\in I$ and $L\ge L_1$,
\[
 \PP\Big\{\argmin_{e\in\V_s}\Rea S_L(e)\not\subseteq\D_s\Big\}\le Ce^{-cL^{d}}.
\]
\item \emph{Gaussian comparison.} Let $X$ be a centered Gaussian vector on $\D_s$ with the exact finite-$L$ covariance $\Gamma_L(v,w)=\operatorname{Cov}(Y_L(v),Y_L(w))$. Then for all $s\in I$ and $L\ge L_1$,
\[
 \Big|\PP(F_L)-\PP\Big\{\min_{v\in\D_s}(\gamma_s+X_v)\le0\Big\}\Big|\le C\Big[\Big(\frac{\norm{A_L}_\op}{\sigma_L}\Big)^{1/3}+e^{-cL^d}\Big]\le CL^{-d/6}.
\]
\item \emph{Limit law.} Let $s_j\uparrow s_*$ and $L_j\to\infty$ with $\gamma_{s_j}/\sigma_{L_j}\to t\in[0,\infty]$. Then
\[
 \PP(F_{L_j})\to H(t):=\PP\Big\{\min_{v\in\D_*}Z_v\le-t\Big\},\qquad H(\infty)=0,
\]
where $\D_*=s_*\D_0$ and $Z$ is centered Gaussian on $\D_*$ with covariance $G$.
\end{enumerate}
\end{theorem}

In Part 1 the constant is the same for every lattice, including lattices whose reciprocal vectors are tangent to the sphere of radius $2k$. Part 2 asserts that the minimizer is a deep hole (Fig.~\ref{fig:geometry}(b)), not merely close to one. The mean grows linearly in the distance to $\D$, with a slope $c_g>0$ determined by $\psi'(R)$ and the cell. The fluctuation gradient satisfies $\sup_\V|\nabla Y_L|=O_P(L^{-d/2})$, and the event $\{\sup_\V|\nabla Y_L|\ge c_g\}$ has probability at most $Ce^{-cL^d}$. Outside that event the linear growth dominates on every segment ending at a deep hole. Part 3 bounds the absolute error by $O(L^{-d/6})$.

The structure of $Z$ is explicit.

\begin{proposition}[Sign classes of deep holes]\label{prop:classes}
$Z_v=Z_{-v}$ almost surely. After one representative of each pair $\{v,-v\}\subset\D_*$ is kept, the covariance matrix is positive definite. The rank of $G$ on $\D_*$ is therefore $|\D_*|/2$: it equals $2$, $3$, $4$, $12$ for the square, hexagonal, simple cubic, and BCC lattices. Lattice-equivalent vertices, such as $v$ and $v-\ell$, are not identified. The function $H$ is continuous and strictly decreasing, and $p_0:=H(0)\in(0,1)$.
\end{proposition}

At the critical radius $R=r_0$, the correlations between distinct classes are $-0.956$ for the square lattice, $-0.49991$ for the hexagonal lattice, and $-0.307$ for the simple cubic lattice. For BCC they take the values $-0.434$, $-0.389$, $-0.249$, $0.006$, and $0.407$. For the square lattice,
\begin{equation}\label{eq:p0square}
 p_0=\frac34-\frac{1}{2\pi}\arcsin\frac{J_0(\sqrt2\,j_{0,1})}{\frac12(1+J_0(2j_{0,1}))}=0.95262\ldots,
\end{equation}
and Monte Carlo with $2\times10^5$ draws from the centered Gaussian distribution with limiting covariance $G$ gives $p_0\approx0.998$ for simple cubic and values that round to $1.000$ for hexagonal and BCC (by Proposition~\ref{prop:classes}, $p_0<1$ in all cases). A single Gaussian would give $1/2$. At finite $L$ the Gaussian vector in Theorem~\ref{thm:localization}.3 keeps all deep holes. The identification $Z_v=Z_{-v}$ holds only in the limit. In the hexagonal case the three classes are nearly, but not exactly, linearly dependent: $1+J_0(2j_{0,1})+2J_0(\sqrt3\,j_{0,1})=1.39\times10^{-4}$.

\begin{theorem}[Critical scaling of the aperture]\label{thm:sharp}
For $\gamma_s>0$ and $\eta\in(0,1)$, let
\[
 L_\eta(\gamma_s)=\inf\{a\ge0:\ \PP(F_L)\le\eta\ \text{for every }L\ge a\}.
\]
For each fixed $\eta\in(0,p_0)$, let $q_\eta>0$ solve $H(q_\eta)=\eta$. Then, as $s\uparrow s_*$,
\[
 L_\eta(\gamma_s)^{d}=\kappa_n\,q_\eta^2\,\gamma_s^{-2}\,(1+o(1)).
\]
For $\eta>p_0$, $L_\eta(\gamma_s)=O(1)$. For $\eta=p_0$, $\gamma_s^2L_\eta(\gamma_s)^d\to0$.
\end{theorem}

The restriction $\eta<p_0$ is necessary for a positive constant: for the critical square lattice $p_0<0.96$, and at $\eta=0.99$ the critical aperture stays bounded as $\gamma\to0$.

Theorem~\ref{thm:sharp} is a critical scaling limit at a fixed target failure probability $\eta<p_0$. The high-reliability regime, in which $\eta\to0$ together with $\gamma$, is governed by the following result. Let
\[
 G_*=\tfrac12\big(1+\psi(2r_0)\big)
\]
be the common limiting variance $G(v,v)$ at the deep holes of the critical lattice, and, for a deep hole $v$ at finite $L$, let $H_v$ be the Hermitian kernel of $Y_L(v)$ (Appendix~\ref{app:kernel}), $\nu_v^2=\tr H_v^2$, $\kappa_{3,v}=2\tr H_v^3$, and $x_v=\gamma_s/\nu_v$.

\begin{theorem}[Joint limit and relative error]\label{thm:joint}
\leavevmode
\begin{enumerate}
\item \emph{Joint limit.} If $\gamma_s\to0$ and $\gamma_s^2L^d\to\infty$, then
\[
 \log\PP(F_L)=-\frac{\gamma_s^2L^d}{2\kappa_nG_*}\,(1+o(1)).
\]
Consequently, for every joint sequence $\gamma_s\downarrow0$ and $\eta\downarrow0$, with no restriction on their relative rates,
\[
 L_\eta(\gamma_s)^d=2\kappa_nG_*\,\gamma_s^{-2}\log(1/\eta)\,(1+o(1)).
\]
\item \emph{Relative error.} If $t_L=\gamma_s/\sigma_L\to\infty$ and $t_L=o(L^{d/6})$, then $\PP(F_L)$ divided by the Gaussian proxy of Theorem~\ref{thm:localization}.3 tends to one. If $t_L\to\infty$ and $t_L=o(L^{d/4}/\sqrt{\log L})$, then
\[
 \PP(F_L)=(1+o(1))\sum_{v}\Phi(-x_v)\exp\Big(-\frac{\kappa_{3,v}\gamma_s^3}{6\nu_v^6}\Big),
\]
where the sum runs over one deep hole from each sign pair.
\end{enumerate}
\end{theorem}

Part 1 gives the exact constant in the order $\gamma^{-2}\log(1/\eta)$ of Theorem~\ref{thm:aperture}. It agrees with Theorem~\ref{thm:sharp} at the boundary between the two regimes: for the limiting Gaussian vector, $H(q)=N_*\Phi(-q/\sqrt{G_*})(1+o(1))$ as $q\to\infty$, where $N_*$ is the number of sign classes, so that
\[
 \frac{q_\eta^2}{G_*}=2w-\log w+2\log N_*-\log(4\pi)+o(1),\qquad w=\log(1/\eta).
\]
Because the deep holes of the critical lattice lie at distance $r_0$ from the origin and $\phi$ is radial, $G_*$ is the same for every lattice. In the high-reliability regime the lattice therefore enters the leading constant only through $\gamma$, that is, through its covering radius. The number of sign classes appears in the quantile of $H$ at the next order. At a fixed target failure probability, by contrast, the whole deep-hole covariance enters $q_\eta$. Transferring that next-order term to the actual aperture requires accuracy conditions on the finite-$L$ covariances that we do not verify here. When $\tr H_v^3>0$, the cubic correction in Part 2 reduces the corresponding Gaussian tail approximation, consistent with the frequencies in Table~\ref{tab:deepholes}.

In Theorems~\ref{thm:localization} and~\ref{thm:sharp} the geometry enters in two separate ways: the covering radius determines $\gamma$, whereas the deep holes determine $q_\eta$ through the covariance $G$ on $\D$, that is, through both the number of sign classes and their relative positions, while $\kappa_n$ does not depend on the lattice. At equal margin $\gamma=0.15$ and $L=8$, the simulated failure frequency of the hexagonal lattice is $15.6$ times that of the square lattice (Section~\ref{sec:numerics}), reflecting differences in both the number of sign classes and their covariance. At equal port density the margin of the hexagonal lattice is much larger, and in the same simulations at $4.01$ points per $\lambda^2$ and $L=4$ it fails in $0.12\%$ of realizations against $66\%$ for the square lattice. Covering-optimal lattices maximize the margin at a given density, and in these simulations they also need the smaller aperture. Theorem~\ref{thm:sharp} applies to one lattice shape at a time and does not by itself compare apertures across shapes.

\begin{proposition}[Other windows and spectra]\label{prop:weight}
For a fixed symmetric convex body $K$, a fixed lattice, and a fixed continuous $f$ with $c_f\le f\le C_f$, the normalized covariance limit in Theorem~\ref{thm:localization}.1 holds with $G$ replaced by
\[
 G_{K,f}(e,w)=\frac{\int f(u)^2\omega_K(u)\cos(ku\cdot e)\cos(ku\cdot w)\,\sigma(\dd u)}{\int f(u)^2\omega_K(u)\,\sigma(\dd u)},\qquad
 \omega_K(u)=\int_{u^\perp}\Big[\int_\R\mathbf 1_K(y+tu)\dd t\Big]^2\dd y .
\]
The weight $\omega_K(u)$ is the integral of the squared chord length of $K$ in direction $u$. If $\argmin_\V m$ is finite and $m$ grows linearly away from it, the localization argument applies to this set. The constants here may depend on $K$, $f$, and the growth constant.
\end{proposition}

The convergence of $L^d\sigma_L^2$ cannot be uniform over the densities with $c_f\le f\le C_f$: for $f_m=1+\frac12\cos(m\theta)$ with $m$ growing faster than $L^2$, $L\sigma_L^2(f_m)-L\sigma_L^2(1)\to0$ while the limits differ by $1/(3\pi^3)$ (Appendix~\ref{app:deepholes}). A rate requires angular regularity of $f$.

\section{Specular Components}\label{sec:specular}

A continuous angular spectrum is what makes $\sigma_L\to0$. A specular path is an atom of the spectrum, and spatial averaging does not remove its random gain. We now let
\[
 \mu=(1-\alpha)f\sigma+\sum_{j=1}^Jp_j\delta_{u_j},\qquad \sum_jp_j=\alpha\in(0,1),
\]
\[
 h(r)=\sqrt{1-\alpha}\,h_c(r)+\sum_{j=1}^J\sqrt{p_j}\,g_j\,e^{\mathrm ik u_j\cdot r},
\]
with $h_c$ a continuous-spectrum field as in Section~\ref{sec:model}, independent of the gains. We consider two gain models:
\begin{itemize}
\item \emph{(G) Rayleigh-faded paths}: $g_j$ i.i.d.\ $\CN(0,1)$, so $h$ is Gaussian.
\item \emph{(D) fixed-amplitude paths}: $g_j=e^{\mathrm i\theta_j}$ with independent phases uniformly distributed on $[0,2\pi)$.
\end{itemize}
Write $\xi_j=ku_j$ and $m_c(e)=(1-\alpha)\Rea\phi_c(e)$, and assume the continuous margin $\gamma_c=\min_\V m_c>0$, except in part~3 of Theorem~\ref{thm:known}, where the lattice is critical and $\gamma_c=0$. Two paths \emph{alias} if $\xi_j-\xi_l\in\G:=2\pi B^{-T}\Z^n$, the angular reciprocal lattice. For an alias class $\mathcal C$ and $e\in\R^n$, let $B_e$ be the Hermitian matrix with entries
\[
 (B_e)_{lj}=\tfrac12\sqrt{p_lp_j}\big(e^{-\mathrm i\xi_l\cdot e}+e^{\mathrm i\xi_j\cdot e}\big)\quad(l,j\in\mathcal C),
\]
assembled block-diagonally over classes. Without aliasing, $B_e=\diag\big(p_j\cos(\xi_j\cdot e)\big)$.

\begin{theorem}[Sign law with Rayleigh-faded paths]\label{thm:specular}
In model (G), for a fixed configuration and large $L$,
\[
 \sup_{e\in\V}|S_L(e)-S_\infty(e)|=O_P\big(L^{-d/2}+L^{-1}\big),\qquad \Rea S_\infty(e)=m_c(e)+g^*B_eg .
\]
Write $g=\sqrt T\,U$, where $T$ has a Gamma distribution with shape $J$ and unit scale and, independently, $U$ is uniform on the unit sphere of $\C^J$, and put
\[
 r(U)=\max\Big\{0,\ \max_{e\in\V}\frac{-U^*B_eU}{m_c(e)}\Big\},\qquad
 H_J(r)=e^{-1/r}\sum_{q=0}^{J-1}\frac{r^{-q}}{q!}\ (r>0),\quad H_J(0)=0 .
\]
Then
\[
 \lim_{L\to\infty}\PP(F_L)=\E\,H_J\big(r(U)\big),
\]
and the limit is positive if and only if $B_e$ has a negative eigenvalue for some $e\in\V$.
\end{theorem}

\begin{corollary}[No power threshold]\label{cor:nothreshold}
Without aliasing, the limit is positive if and only if $\cos(\xi_j\cdot e)<0$ for some path $j$ and some $e\in\V$, whatever the powers $p_j>0$. For a single path of power $p$,
\[
 \lim_{L\to\infty}\PP(F_L)=\exp\Big(-\frac{(1-p)\kappa}{p}\Big),\qquad
 \kappa=\inf_{e\in\V:\ \cos(\xi\cdot e)<0}\frac{\Rea\phi_c(e)}{-\cos(\xi\cdot e)} ,
\]
with $\kappa=\infty$ and limit $0$ when $\cos(\xi\cdot e)\ge0$ on $\V$.
\end{corollary}

When some $e\in\V$ has $\cos(\xi\cdot e)<0$, the floor is positive for every $p>0$ but exponentially small in $1/p$, and it is at most $\eta$ if and only if
\[
 p\le p_\eta:=\frac{\kappa}{\kappa+\log(1/\eta)} .
\]
This is the effective power threshold at target failure probability $\eta$. When $\kappa=\infty$ the floor vanishes and $p_\eta=1$. With $\Rea\phi_c\approx0.2$ and $|\cos|\approx0.5$ where the cosine of the path phase change is negative, $\kappa\approx0.4$. A path carrying $10\%$ of the power leaves a floor of about $e^{-3.6}\approx0.027$, and $p_\eta\approx0.055$, about $-12.6$ dB relative to the total average power, for $\eta=10^{-3}$.

\begin{proposition}[Fixed-amplitude paths]\label{prop:fixed}
In model (D) without aliasing, $S_\infty(e)=(1-\alpha)\phi_c(-e)+\sum_jp_je^{-\mathrm i\xi_j\cdot e}$ is deterministic. Hence $\PP(F_L)\to0$ if $\min_\V\Rea S_\infty>0$ and $\PP(F_L)\to1$ if $\min_\V\Rea S_\infty<0$. For a single path, with $\kappa$ as above,
\[
 \min_\V\Rea S_\infty>0\iff p<p^*:=\frac{\kappa}{1+\kappa},
\]
where $p^*=1$ if $\kappa=\infty$.
\end{proposition}

For a fixed-amplitude path, failures disappear as $L\to\infty$ below the power threshold $p^*$. For a Rayleigh-faded path there is no such threshold: wherever the cosine of its phase change is negative in the cell, a sufficiently large gain draw reverses the sign, and this event has a probability that does not depend on $L$. With aliasing in model (D), the relative phases $\theta_j-\theta_l$ enter $S_\infty$ and the limit is random again.

\begin{corollary}[Robustness to every direction]\label{cor:quarter}
$\cos(ku\cdot e)\ge0$ for all unit vectors $u$ and all $e\in\V$ if and only if $R\le\lambda/4$. In that case distinct directions do not alias, and in both gain models $\PP(F_L)\to0$ for every finite configuration of paths with $\gamma_c>0$. The least densities are those of the second row of Table~\ref{tab:densities}: $6.158$ per $\lambda^2$ (hexagonal) and $22.36$ per $\lambda^3$ (BCC).
\end{corollary}

Robustness to an unknown specular direction thus costs a factor $(r_0/(\lambda/4))^n$ in density: $2.34$ in the plane and $8$ in space. When the direction is known, the cost is much smaller.

\begin{theorem}[Known direction in the plane]\label{thm:known}
Let $n=2$, $f\equiv1$, and let one path have known direction $u=(1,0)$. Let $\rho=r_0$, $h=1/4$, $s=\sqrt{\rho^2-h^2}$, and
\[
 C=\{e:\ |e|\le\rho,\ |e_1|\le h\},\qquad
 B_*=\begin{pmatrix}2h&h\\0&\rho+s\end{pmatrix},\qquad \Lambda_*=B_*\Z^2 .
\]
\begin{enumerate}
\item For every lattice with $\V(\Lambda)\subset C$, and also for every lattice with $\Lambda+C=\R^2$ (an arbitrary measurable remap rule), $\covol\Lambda\le2h(\rho+s)$. Equality holds for $\Lambda_*$, whose Voronoi cell is the hexagon with vertices $(0,\pm\rho)$, $(\pm h,\pm s)$. The least density is
\[
 \delta_*=\frac{1}{2h\big(\rho+\sqrt{\rho^2-h^2}\big)}=2.97375\ldots\ \text{per }\lambda^2 .
\]
\item The infimum is not attained under a strictly positive uniform margin. The lattices $(1-\varepsilon)\Lambda_*$ have positive continuous and specular margins and satisfy $\PP(F_L)\to0$ for each $\varepsilon\in(0,1)$.
\item For $\Lambda_*$ itself, in model (G), $\PP(F_L)\ge1/36$ for every finite nonempty aperture.
\end{enumerate}
\end{theorem}

The geometric statement in part~1 holds with $\rho$ replaced by any $R'>0$ and any $0<h\le\sqrt3R'/2$. Parts~2 and~3 use $\rho=r_0$ and $h=1/4$. The optimal lattice is not a hexagonal lattice: its Voronoi cell is the largest centrally symmetric hexagon inscribed in the disk-slab intersection $C$. A known direction costs $13\%$ in density relative to diffuse scattering ($2.974$ against $2.627$), compared with $134\%$ for an unknown direction.

\begin{remark}
Theorem~\ref{thm:known} restricts residuals to the component of $\{J_0(2\pi|e|)>0\}$ containing the origin. The full positive set also contains annuli beyond the second zero. Remap rules that use them admit lattices $\frac12\Z\times(N+\sqrt2)\Z$ of density $2/(N+\sqrt2)\to0$ with positive mean on a fundamental domain. Their margins on these domains are of order $N^{-1/2}$, and their Euclidean Voronoi cells have no positive margin. These lattices lie outside the compact family of Theorem~\ref{thm:aperture}, and their aperture requirement is not determined here.
\end{remark}

\section{Numerical Comparison}\label{sec:numerics}

We compare the deep-hole law with independent simulations of the field. The simulations use an exact angular Fourier--Bessel representation of the isotropic planar field, not a finite sum of plane waves. A finite plane-wave sum is itself an atomic spectrum and would mix in the effects of Section~\ref{sec:specular}. The angular expansion is truncated so that the omitted variance is below $10^{-14}$ at every point within distance $L+R$ of the origin, which covers all samples and all shifted points $a+e$. For each realization, the minimum of $\Rea S_L$ over the closed cell is bracketed, with maximal bracket width about $10^{-5}$. The cell is triangulated. On each triangle the lower bound is the smallest vertex value minus $H\operatorname{diam}^2/6$, where $H$ bounds the Hessian norm on the cell, and the upper bound comes from evaluated points and constrained local optimization. Triangles are subdivided until the bracket width is at most $10^{-5}$, and margins for the truncation and for floating point are then added. The code and the bracket endpoints of every realization are available at \url{https://doi.org/10.5281/zenodo.23687425}~\cite{zenodo}. The bracket is for the truncated model in floating point. The predictions are the deep-hole Gaussian proxy of Theorem~\ref{thm:localization}.3, with the exact finite-$L$ mean and covariance, and, for contrast, the same proxy at a single deep hole.

\begin{table}[t]
\centering
\caption{Failures over the whole cell out of $4096$ independent realizations, each classified by the bracket (no bracket contained zero in these rows), against the deep-hole Gaussian proxy and the single-point proxy, each evaluated with $4\times10^5$ Gaussian draws. Ball window of radius $L$, isotropic planar field.}
\label{tab:deepholes}
\begin{tabular}{llrrrcrr}
\toprule
lattice & case & $L$ & $M$ & failures & $95\%$ Wilson interval & deep holes & single point\\
\midrule
square & $\gamma=0.15$ & 4 & 221 & 500 & $[0.112,0.132]$ & 0.146 & 0.046\\
hexagonal & $\gamma=0.15$ & 4 & 163 & 2556 & $[0.609,0.639]$ & 0.636 & 0.137\\
square & $\gamma=0.15$ & 8 & 877 & 53 & $[0.0099,0.0169]$ & 0.0162 & 0.0057\\
hexagonal & $\gamma=0.15$ & 8 & 673 & 827 & $[0.190,0.214]$ & 0.225 & 0.041\\
hexagonal & $\gamma=0.15$ & 16 & 2713 & 34 & $[0.0059,0.0116]$ & 0.0110 & 0.0020\\
square & density $4.01$ & 4 & 197 & 2714 & $[0.648,0.677]$ & 0.668 & 0.197\\
hexagonal & density $4.01$ & 4 & 199 & 5 & $[0.0005,0.0029]$ & 0.0056 & 0.0012\\
\bottomrule
\end{tabular}
\end{table}

Table~\ref{tab:deepholes} bears on localization, on the role of the deep-hole configuration, and on the accuracy of the Gaussian proxy in the tail.
\begin{itemize}
\item \emph{Localization.} In all $28{,}672$ realizations, the approximate minimizer found over the cell was a deep hole, and within the truncation and the numerical tolerance the failure classification over the cell and over the deep holes never disagreed.
\item \emph{Multiplicity.} When the failure probability is moderate, the multi-hole proxy is close to the simulation, and the single-point proxy is off by a factor of $2$ to $5$. At equal margin $\gamma=0.15$, the simulated failure frequency of the hexagonal lattice, with three sign classes, is $15.6$ times that of the square lattice, with two ($L=8$).
\item \emph{Tails.} The Gaussian proxy is larger than the simulated frequency in every row, and the ratio grows as the probability decreases, reaching $4.6$ for the point estimate in the last row. This is the direction predicted by the cubic factor of Theorem~\ref{thm:joint}.2. For the last row, the cubic-corrected sum over sign classes, evaluated with the exact finite-$L$ spectra, is $0.00059$, inside the Wilson interval, while the Gaussian proxy is $0.0056$. Here $\gamma/\sigma_L\approx1.7$ is not large, so the asymptotic expansion provides no numerical error guarantee and this agreement is only indicative.
\end{itemize}

The convergence of the normalization is fast for some lattices and slow for others. With density $4$ per $\lambda^2$, $L\sigma_L^2=0.0861$ for the hexagonal lattice at $L=32$, already at $\kappa_2=0.0860$. For the square lattice, whose reciprocal lattice has vectors of length $2k$, it is $0.0959$. In three dimensions at density $8$, $L^2\sigma_L^2$ is $0.0285$ for BCC and $0.0294$ for simple cubic at $L=16$, against $\kappa_3=0.0285$. In all cases the slope of $\log\sigma_L^2$ against $\log L$ is close to $-d$ ($-1.005$ hexagonal, $-1.075$ square, $-2.021$ BCC, $-2.073$ simple cubic, least-squares slopes in log-log coordinates over $L\in\{8,16,32\}$ in the plane and $L\in\{4,8,16\}$ in space), not $-n$.

The specular floor of Corollary~\ref{cor:nothreshold} is approached slowly when the continuous margin is small. With one Rayleigh-faded path ($p=0.1$) along $(1,0)$ and the rectangular cell $[-0.3,0.3]\times[-0.2,0.2]$, the limit is $0.115$. The simulated failure frequency decreases as $0.737$, $0.593$, $0.432$, $0.301$ for $L=4,8,16,32$. The continuous margin here is only $0.067$, so the diffuse fluctuation dominates up to large apertures. For the cell $[-0.24,0.24]\times[-0.2,0.2]$, which lies inside the slab $|e_1|<1/4$, the frequency is $0$ out of $2048$ at $L=8$ and $L=16$, consistent with a zero floor.

\section{Scope and Open Problems}\label{sec:discussion}

\emph{Approximation in the tail.} Theorem~\ref{thm:localization}.3 bounds the absolute error of the Gaussian proxy, and Theorem~\ref{thm:joint} gives relative error and the exact logarithmic rate near the critical lattice. Two questions remain: the second-order term of the actual aperture in the joint limit, for which the number of sign classes enters the limiting quantile through $2\log N_*$, and large deviations at a fixed positive margin, where the localization estimate of Theorem~\ref{thm:localization}.2 does not by itself transfer a single-point rate to the whole cell.

\emph{Worst versus average residual.} For isotropic scattering $m(e)=1-k^2|e|^2/(2n)+O(|e|^4)$. If the residual is uniform on the cell and the cell is small compared with the wavelength, the average of $m$ over $\V$ is therefore governed to leading order by the normalized second moment of $\V$, which leads to the lattice quantizer problem~\cite{conway,barnes}. For a general spectrum the quadratic term is weighted by $\int uu^\top\mu(\dd u)$. At fixed density, the isotropic criterion selects the same lattice shapes as covering in two and three dimensions. The worst residual, studied here, brings in the deep holes.

\emph{Several outputs.} The reduction of Lemma~\ref{lem:identity} uses a single receive antenna, so that the error signal is a scalar. With $N_r$ outputs, $\Rea\ip g{\widehat g}=M\,\Rea(\Delta^*\Sigma_L(e)\Delta)$, where $\Delta$ is the error vector and $\Sigma_L(e)$ is the empirical cross-correlation matrix of the $N_r$ fields. The descent condition holds simultaneously for all nonzero error vectors and all $e\in\V$ if and only if the Hermitian part of $\Sigma_L(e)$ is positive definite for every $e\in\V$. For independent output fields its mean is $m(e)I$.

\emph{Parameterized transmitters.} The statement concerns the gradient with respect to the transmit vector. For a front end $x=P(w)$ the returned gradient is composed with the adjoint of the Jacobian of $P$, which can change angles. Return noise is treated in Proposition~\ref{prop:noise}.

\emph{Rigid translation and the field model.} The residual is the same at every port because the array translates rigidly. A rotation by a small angle $\theta$ adds a port-dependent residual of size up to $\theta L$. The model does not cover it, and it is small compared with the cell only when $\theta L\ll R$. The correlation model ignores mutual coupling between ports. At $R=r_0$ the spacings in Table~\ref{tab:densities} are between $0.54\lambda$ and $0.78\lambda$. At $R=\lambda/4$ they are between $0.29\lambda$ and $0.43\lambda$, below $\lambda/2$. At these spacings mutual coupling can affect the port correlations and may require an array-dependent channel model.

\emph{Further problems.} Three questions remain open: the optimal three-dimensional lattice for a known specular direction (Theorem~\ref{thm:known} relies on F\'ary's planar theorem), the rate at which $\PP(F_L)$ approaches the specular floor, which Section~\ref{sec:numerics} shows can be slow when the continuous margin is small, and the first correction to $L^d\sigma_L^2\to\kappa_n$, which depends on the reciprocal lattice near the sphere of radius $2k$.

\appendices
\section{Quadratic-Form Representation and Kernel Bounds}\label{app:kernel}

Throughout, $\A=A_L^\ell$, $M=|\A|$, $v_a(u)=e^{\mathrm ika\cdot u}$ on $L^2(\mu)$, and
\[
 A=\frac1M\sum_{a\in\A}|v_a\rangle\langle v_a|,\qquad (U_e z)(u)=e^{\mathrm iku\cdot e}z(u),\qquad H_e=\tfrac12(U_eA+AU_e^*).
\]

\begin{lemma}[Second moments]\label{lem:wick}
$\Rea S_L(e)$ is the Hermitian Gaussian quadratic form associated with $H_e$, and $\tr H_e=m(e)$. For any finitely many lags, with $Y(e)=\Rea S_L(e)-m(e)$,
\[
 \operatorname{Cov}(Y(e),Y(w))=\tr(H_eH_w),\qquad \operatorname{Var}Y(e)=\norm{H_e}_\HS^2,\qquad \sigma_L^2=\norm{A}_\HS^2 .
\]
The kernel of $H_e$ is
\begin{equation}\label{eq:Hkernel}
 H_e(u,u')=D_L\big(k(u-u')\big)\,e^{\mathrm ik(u-u')\cdot e/2}\cos\big(k(u+u')\cdot e/2\big),\qquad D_L(q)=\frac1M\sum_{a\in\A}e^{\mathrm ia\cdot q}.
\end{equation}
Consequently $\operatorname{Var}Y(e)\le\sigma_L^2$, and $\norm{H_e}_\op\le\norm{A}_\op$.
\end{lemma}

\begin{proof}
Let $W$ be a proper complex isonormal Gaussian process on $L^2(\mu)$, that is, $\E[W(\varphi)\overline{W(\chi)}]=\ip{\chi}{\varphi}$ and $\E[W(\varphi)W(\chi)]=0$, and put $h(x)=W(v_x)$. Its covariance is~\eqref{eq:phi} and its pseudo-covariance is zero. On the finite span of $\{v_a,v_{a+e}\}$, choose an orthonormal basis. The coordinates $z$ of $W$ are i.i.d.\ $\CN(0,1)$ and $\Rea S_L(e)=z^*H_ez$. For such $z$, $\E[z_i\bar z_jz_p\bar z_q]=\delta_{ij}\delta_{pq}+\delta_{iq}\delta_{pj}$, which gives the covariance identities. The kernel follows from $U_eA+AU_e^*$. The bounds follow from $|\cos|\le1$ and $\norm{U_e}_\op=1$.
\end{proof}

\begin{lemma}[Operator bound]\label{lem:op}
In the family~\eqref{eq:family} and for $L\ge L_0$, $\norm{A}_\op\le C_*L^{-d}$, where
\[
 C_*=e^{r_{\rm out}^2/2}(2\pi)^{n/2}\frac{2^n}{\omega_nr_{\rm in}^n}\frac{C_f2^d\Gamma(1+d/2)}{k^d}Q_n(g),\quad
 Q_n(g)=\sum_{j=0}^n\binom nj\Big(1+\frac{4k}{g}\Big)^{n-j}\Big(\frac4g\Big)^j\Gamma\Big(1+\frac j2\Big),
\]
$g=k/b_{\max}$, and $\omega_n$ is the volume of the unit ball. Hence $\sigma_L^2\le\norm A_\op\tr A\le C_*L^{-d}$.
\end{lemma}

\begin{proof}
\emph{Counting.} The parallelotope $P=B[-\frac12,\frac12)^n$ lies in $B_{R_*}$. If $|a|\le Lr_{\rm in}-T$, then $a,a-\ell\in LK$. Tiling $B_{Lr_{\rm in}-T-R_*}$ by translates of $P$ gives $Mv\ge\omega_n(Lr_{\rm in}/2)^n$ for $L\ge L_0$. Also $|a|\le Lr_{\rm out}$ for $a\in\A$.

\emph{Gaussian majorant.} Since $e^{-|a|^2/(2L^2)}\ge e^{-r_{\rm out}^2/2}$ on $\A$, in positive order
\[
 A\preceq\widetilde A=\frac{e^{r_{\rm out}^2/2}}{M}\sum_{a\in\Lambda}e^{-|a|^2/(2L^2)}|v_a\rangle\langle v_a| .
\]
Poisson summation for the Gaussian~\cite[Ch.~5, Thm.~3.1]{stein} gives
\[
 \sum_{a\in\Lambda}e^{-|a|^2/(2L^2)}e^{\mathrm ia\cdot q}=\frac{(2\pi)^{n/2}L^n}{v}\sum_{G\in\G}e^{-L^2|q-G|^2/2}.
\]

\emph{Row integrals.} Let $q=k(u-u')$, so $|q|\le2k$. For $L\ge1$, $e^{-L^2|q-G|^2/2}\le e^{-(|G|-2k)_+^2/4}e^{-(kL)^2|u'-(u-G/k)|^2/4}$. For every center $c$ and radius $r$, $\sigma(S^{n-1}\cap B_r(c))\le r^d$, so $\mu(B_r(c))\le C_fr^d$. Writing the Gaussian as an integral of indicator functions of balls,
\[
 \int e^{-(kL)^2|u'-c|^2/4}\mu(\dd u')\le\frac{C_f2^d\Gamma(1+d/2)}{(kL)^d}.
\]
Reciprocal lattice points are $g$-separated, so the number in a ball of radius $s$ is at most $(1+2s/g)^n$, and $\sum_Ge^{-(|G|-2k)_+^2/4}\le Q_n(g)$. Combining these with the counting bound, every row and column integral of $|\widetilde A(u,u')|$ is at most $C_*L^{-d}$. The Schur test gives $\norm{\widetilde A}_\op\le C_*L^{-d}$. No step divides by $\big||G|-2k\big|$, so tangencies are included.
\end{proof}

\begin{lemma}[Real-part non-degeneracy]\label{lem:nondeg}
In the family~\eqref{eq:family}, $\operatorname{Var}\Rea S_L(e)\ge c_*L^{-d}$ for $e\in\V$ and $L\ge L_0$, with
\[
 c_*=\frac{c_f^2p_*\alpha_n\varepsilon_*^d}{8},\quad a_*=\min\{1,(12R_*)^{-1}\},\quad p_*=\frac{2a_*}{\pi},\quad \varepsilon_*=\min\Big\{1,\frac1{2kr_{\rm out}},\frac1{12R_*}\Big\},
\]
$\alpha_2=1/\pi$, $\alpha_3=1/4$.
\end{lemma}

\begin{proof}
Let $E_e=\{u:|u\cdot e|\le1/12\}$. It contains an equatorial band of $\sigma$-mass at least $p_*$. If $|u-u'|\le\varepsilon_*/L$, then $|D_L(k(u-u'))-1|\le kr_{\rm out}\varepsilon_*\le\frac12$, and for $u\in E_e$, $|k(u+u')\cdot e/2|\le\pi/6+kR_*\varepsilon_*/(2L)\le\pi/4$. In~\eqref{eq:Hkernel} this gives $|H_e|^2\ge\frac14\cdot\frac12$ on that set. Since $f\ge c_f$ and a cap of radius $r\le1$ centered on the sphere has $\sigma$-mass at least $\alpha_nr^d$,
\[
 \norm{H_e}_\HS^2\ge\frac{c_f^2}{8}\int_{E_e}\sigma\big(B_{\varepsilon_*/L}(u)\cap S^{n-1}\big)\sigma(\dd u)\ge c_*L^{-d}.
\]
Only pointwise bounds on $f$ are used.
\end{proof}

\begin{lemma}[Increments]\label{lem:incr}
$\norm{H_e-H_{w}}_\HS\le k|e-w|\sigma_L$ and $\norm{H_e-H_{w}}_\op\le k|e-w|\norm A_\op$. More generally, for every multi-index $\alpha$, $\norm{\partial_e^\alpha H_e}_\HS\le k^{|\alpha|}\sigma_L$ and $\norm{\partial_e^\alpha H_e}_\op\le k^{|\alpha|}\norm A_\op$.
\end{lemma}

\begin{proof}
$\norm{U_e-U_w}_\op\le k|e-w|$, and $\partial^\alpha_eU_e$ is multiplication by $(\mathrm iku)^\alpha e^{\mathrm iku\cdot e}$, of modulus at most $k^{|\alpha|}$. Apply these to $H_e=\frac12(U_eA+AU_e^*)$.
\end{proof}

\begin{lemma}[Chaos tail]\label{lem:tail}
Let $X=z^*Hz-\tr H$ with $H$ Hermitian of finite rank, $z$ standard proper complex Gaussian, $s^2=\norm H_\HS^2$ and $b=\norm H_\op$. Then $X=\sum_j\lambda_j(E_j-1)$ with $E_j$ i.i.d.\ unit exponentials, and
\[
 \PP\big\{|X|>s\sqrt{2u}+bu\big\}\le2e^{-u}\qquad(u\ge0).
\]
Moreover $\E X^4\le9s^4$.
\end{lemma}

\begin{proof}
Diagonalize. For $|\theta|b<1$, $\log\E e^{\theta X}=\sum_j[-\log(1-\theta\lambda_j)-\theta\lambda_j]\le\theta^2s^2/(2(1-|\theta|b))$. Take $\theta=\sqrt{2u}/(s+b\sqrt{2u})$ in Markov's inequality, and repeat for $-X$. For the fourth moment, $\E(E-1)^4=9$ and independence give $\E X^4=3s^4+6\sum_j\lambda_j^4\le9s^4$.
\end{proof}

\section{Proof of Theorem~\ref{thm:aperture} and Proposition~\ref{prop:log}}\label{app:aperture}

Proposition~\ref{prop:scale} follows from Lemmas~\ref{lem:wick}, \ref{lem:op}, and~\ref{lem:nondeg}. Put
\[
 D_*=1+3kR_*,\qquad h_n=2n\log3+4(n+1)\log2,\qquad C_3=8C_*D_*^2h_n,\qquad c_4=(c_*/144)^{1/d}.
\]

\emph{Uniform tail over the cell.} Let $X_e=Y(e)=\Rea S_L(e)-m(e)$. Its paths are continuous. Let $N_0=\{0\}$ and, for $j\ge1$, let $N_j$ be a maximal $R_*2^{-j}$-separated subset of $\V\subset B_{R_*}$, so $|N_j|\le(3\cdot2^j)^n$. Choose $\pi_j(e)\in N_j$ within $R_*2^{-j}$ of $e$, with $\pi_0(e)=0$. Then $|\pi_j(e)-\pi_{j-1}(e)|\le3R_*2^{-j}$, and there are at most $9^n2^{2nj}$ such pairs at level $j$. Apply Lemma~\ref{lem:tail} to $X_0$ with $u_0=t+\log4$ and to each increment with
\[
 u_j=t+2n\log3+\log4+(2n+1)j\log2 .
\]
By Lemma~\ref{lem:incr}, an increment over a pair at level $j$ has HS norm at most $3kR_*2^{-j}\sigma_L$ and operator norm at most $3kR_*2^{-j}\norm{A}_\op$. The union of all exceptional events has probability at most $e^{-t}/2+\frac{e^{-t}}2\sum_j2^{-j}=e^{-t}$. Outside it, telescoping, continuity, $\sum_j2^{-j}u_j=t+h_n$, and Cauchy--Schwarz give
\begin{equation}\label{eq:uniformtail}
 \PP\Big\{\sup_{e\in\V}|X_e|>D_*\big[\sigma_L\sqrt{2(t+h_n)}+\norm A_\op(t+h_n)\big]\Big\}\le e^{-t}\qquad(t\ge0).
\end{equation}
The net is fixed by $\V$ and does not depend on $\gamma$.

\emph{Part 1.} Since $m(0)=1$, $0<\gamma\le1$. Let $t=\log(1/\eta)$ and $x=L^d$. The assumption and $h_n>1$ give $C_*(t+h_n)/x\le\gamma^2/(8D_*^2)$. By Lemma~\ref{lem:op}, the threshold in~\eqref{eq:uniformtail} is at most
\[
 D_*\Big[\sqrt{2C_*(t+h_n)/x}+C_*(t+h_n)/x\Big]\le\frac\gamma2+\frac{\gamma^2}{8D_*}\le\frac{5\gamma}8<\gamma .
\]
On $F_L$, $\sup_\V|X_e|\ge\gamma$. Hence $\PP(F_L)\le\eta$.

\emph{Part 2.} Let $e^*$ minimize $m$ on $\V$, let $X=X_{e^*}$ and $s^2=\E X^2$. Lemma~\ref{lem:tail} gives $\E X^4\le9s^4$. H\"older then gives $\E|X|\ge s/3$. Since $\E X=0$, the negative part $X_-$ has $\E X_-\ge s/6$ and $\E X_-^2\le s^2$. For $a\ge0$, $\PP(X_-\ge a)\ge(\E X_--a)_+^2/\E X_-^2$, which is at least $1/144$ at $a=s/12$. If $L\le c_4\gamma^{-2/d}$, then by Lemma~\ref{lem:nondeg} $s^2\ge c_*L^{-d}\ge144\gamma^2$, so $\gamma\le s/12$ and $\PP(F_L)\ge\PP(X\le-\gamma)\ge1/144$.

\emph{Part 3.} If $\gamma=0$, then $\E X_-\le\sqrt{\E X_-^2\,\PP(X<0)}$ gives $\PP(X<0)\ge(s/6)^2/s^2=1/36$. If $\gamma<0$, Chebyshev's inequality at $e^*$ gives $\PP(\Rea S_L(e^*)>0)\le s^2/\gamma^2\le C_*L^{-d}/\gamma^2$. \hfill$\square$

\emph{Proof of Proposition~\ref{prop:log}.} Take $\Lambda_a=a\Z^3$, a cubic window, and $f\equiv1$. Then $R=\sqrt3a/2$ and $\gamma(a)=\sin(\pi\sqrt3a)/(\pi\sqrt3a)$, which decreases to $0$ as $a\uparrow1/\sqrt3$. Fix $q>0$ and choose $a_L$ with $\gamma(a_L)=\sqrt q/L$. By Lemmas~\ref{lem:op} and~\ref{lem:nondeg}, the eigenvalues of $H_{e^*}$ satisfy $\max_j|\lambda_j|/s=O(L^{-1})$, so $\sum_j\E|\lambda_j(E_j-1)/s|^3\le C\max_j|\lambda_j|/s\to0$. The Lindeberg--Feller theorem~\cite[Thm.~3.4.10]{durrett} applied to $\sum_j\lambda_j(E_j-1)/s$ gives
\[
 \liminf_L\PP(F_L)\ge\liminf_L\PP\big(X/s\le-\sqrt q/(sL)\big)\ge\Phi\big(-\sqrt{q/c_*}\big)>0,
\]
while $\gamma^2L^d=q$. Suppose $L^d\ge C\gamma^{-2}[1+g(\log1/\eta)]$ always implied $\PP(F_L)\le\eta$, with $g(t)=o(t)$. Fix a large $t$ and set $\eta=e^{-t}$ and $q=\max\{1,C[1+g(t)]\}$. Then the implication applies, but the Gaussian lower tail gives $\Phi(-\sqrt{q/c_*})>e^{-t}$ for $t$ large, since $q=o(t)$. This is a contradiction. \hfill$\square$

\section{Proofs for Section~\ref{sec:deepholes}}\label{app:deepholes}

Constants depend only on $n,\Lambda_0,s_-$. Write $\sigma_L$, $A$, $H_e$ for the dilated lattice $\Lambda_s$. All bounds below are uniform in $s\in I$. Let $a_n=|B_1^n|$ and $s_{n-1}=na_n$.

\subsection{Normalization and covariance (Theorem~\ref{thm:localization}.1)}\label{app:normcov}

\emph{Smooth window.} Write $\Lambda_0=B_0\Z^n$, $v_s=\covol(\Lambda_s)$, and $\G_s=2\pi(sB_0)^{-T}\Z^n$. For radial $w\in C_c^\infty$ let $\widehat w(\xi)=\int w(x)e^{\mathrm ix\cdot\xi}\dd x$ and $D^w_L(q)=M^{-1}\sum_{a\in\Lambda_s}w(a/L)e^{\mathrm ia\cdot q}$. Poisson summation~\cite[Ch.~5, Thm.~3.1]{stein} gives
\[
 D^w_L(q)=\frac{L^n}{v_sM}\sum_{G\in\G_s}\widehat w\big(L(q-G)\big).
\]
Since $\widehat w$ decays rapidly and $\G_s$ is uniformly separated, at most one term is not $O(L^{-N})$ at each $q$, so $|D^w_L|^2=\big(\frac{L^n}{v_sM}\big)^2\sum_G|\widehat w(L(q-G))|^2+O(L^{-N})$ uniformly on $|q|\le2k$.

\emph{Aliases.} For $U,U'$ independent uniform on the sphere, $q=k(U-U')$ has density $p_2(q)=\big(\pi^2|q|\sqrt{4k^2-|q|^2}\big)^{-1}$ on $0<|q|<2k$ for $n=2$, and $p_3(q)=(8\pi k^2|q|)^{-1}$ for $n=3$. Away from the origin, the mass of a ball of radius $r$ is $O(r^{3/2})$ for $n=2$, because of the inverse square-root singularity at $|q|=2k$, and $O(r^3)$ for $n=3$. Hence each $G\ne0$ contributes
\[
 \iint|\widehat w(L(k(u-u')-G))|^2\sigma(\dd u)\sigma(\dd u')=O(L^{-3/2})\ (n=2),\qquad O(L^{-3})\ (n=3),
\]
which is $o(L^{-d})$. This holds uniformly whether $|G|-2k$ is positive, zero, or negative.

\emph{Diagonal term.} For continuous $\Psi$ on the product of two spheres, write $u'=\sqrt{1-|y|^2}\,u+y$ with $y\in u^\perp$ and $x=kLy$. Dominated convergence gives
\[
 L^d\iint|\widehat w(kL(u-u'))|^2\Psi(u,u')\dd\sigma\dd\sigma\to\frac{1}{k^ds_{n-1}}\int_{\R^d}|\widehat w(x,0)|^2\dd x\int\Psi(u,u)\sigma(\dd u).
\]
By Fubini and Plancherel, $\int_{u^\perp}|\widehat w|^2=(2\pi)^d\int_{u^\perp}\big[\int w(y+tu)\dd t\big]^2\dd y$.

\emph{Hard window.} Let $w_\delta$ be radial, $0\le w_\delta\le1$, equal to $1$ on $B_{1-\delta}$ and $0$ outside $B_1$. Let
\[
 A^w=\frac1M\sum_{a\in\Lambda_s}w(a/L)\,|v_a\rangle\langle v_a| ,
\]
which keeps the hard-window normalization $M$. Then $0\le A-A^{w_\delta}\le A$ and $\tr(A-A^{w_\delta})$ is $M^{-1}$ times a weighted count of the lattice points with $(1-\delta)L\le|a|\le L$. The fundamental cells of these points, of diameter at most $C_P$ uniformly in $s\in I$, lie in the shell $(1-\delta)L-C_P\le|x|\le L+C_P$. Hence their number is at most $v_s^{-1}$ times its volume, that is $O(\delta L^n+L^{n-1})$ uniformly in $s$, and $\tr(A-A^{w_\delta})\le C(\delta+L^{-1})$. Only this upper bound is needed, not the asymptotics of the count. Lemma~\ref{lem:op} then gives $\norm{A-A^{w_\delta}}_\HS^2\le C(\delta+L^{-1})L^{-d}$. By Cauchy--Schwarz, replacing $D_L$ by $D_L^{w_\delta}$ changes $L^d\iint|D_L|^2\Psi$ by at most $C\norm\Psi_\infty\sqrt{\delta+L^{-1}}$. Letting $L\to\infty$ and then $\delta\to0$, and using $L^n/(v_sM)\to a_n^{-1}$ and the chord function $2\sqrt{1-|y|^2}$ of the ball,
\begin{equation}\label{eq:weak}
 L^d\iint|D_L(k(u-u'))|^2\Psi(u,u')\dd\sigma\dd\sigma\to\kappa_n\int\Psi(u,u)\sigma(\dd u),\qquad
 \kappa_n=\frac{8a_{n-1}}{(n+1)a_n^2s_{n-1}} .
\end{equation}
This gives $\kappa_2=8/(3\pi^3)$ and $\kappa_3=9/(32\pi^2)$. The convergence is uniform over bounded equicontinuous families of $\Psi$. By~\eqref{eq:Hkernel}, $\operatorname{Cov}(Y(e),Y(w))=\tr(H_eH_w)$ is~\eqref{eq:weak} with a test whose diagonal is $\cos(ku\cdot e)\cos(ku\cdot w)$. Dividing by the case $\Psi\equiv1$ gives $G(e,w)=\int\cos(ku\cdot e)\cos(ku\cdot w)\dd\sigma=\frac12[\phi(e-w)+\phi(e+w)]$. For a general window $K$ the same computation, with chord function $\int\mathbf 1_K(y+tu)\dd t$, gives Proposition~\ref{prop:weight}. Since $\int\cos^2(ku\cdot e)\dd\sigma$ has a positive minimum on compact sets,
\begin{equation}\label{eq:nondeg2}
 c\sigma_L^2\le\operatorname{Var}Y(e)\le\sigma_L^2,\qquad cL^{-d}\le\sigma_L^2\le CL^{-d}.
\end{equation}

\subsection{Exact localization (Theorem~\ref{thm:localization}.2)}

\emph{Linear gap.} Let $E_0$ be the vertex set of $\V_0$ and, for $v\in\D_0$, let $a_v=\frac12\min_{w\in E_0\setminus\{v\}}|w-v|$. For $w\in E_0$, $|w|\le|v|$ gives $v\cdot(w-v)\le-a_v|w-v|$. By convexity $v\cdot(e-v)\le-a_v|e-v|$ for $e\in\V_0$, hence $R_0-|e|\ge\frac{a_v}{2R_0}|e-v|$ near $v$. By compactness elsewhere, $R_0-|e|\ge c_\V\dist(e,\D_0)$ on $\V_0$. Since $\psi$ is strictly decreasing on $[0,r_0]$ with $\psi'<0$ on $(0,r_0]$, the divided difference $(\psi(r)-\psi(R))/(R-r)$, extended by $-\psi'(R)$ at $r=R$, is continuous and positive on $\{0\le r\le R,\ s_-R_0\le R\le r_0\}$, hence bounded below. Therefore
\begin{equation}\label{eq:gap}
 m(e)-\gamma_s\ge c_g\dist(e,\D_s)\qquad(e\in\V_s,\ s\in I).
\end{equation}

\emph{Gradient tail.} By Lemma~\ref{lem:incr} and the moment form of Lemma~\ref{lem:tail}, $\norm{\partial^\alpha Y(e)}_p\le3k^{|\alpha|}(\sigma_L\sqrt p+\norm A_\op p)$ for even $p$. On a fixed cube $Q_0\supset\V_{s_*}$ of side $\ell_0$, the averaged fundamental theorem of calculus in each coordinate gives $\sup_{Q_0}|g|\le\sum_{J}\ell_0^{-(n-|J|)}\int_{Q_0}|\partial_Jg|$, where $J$ ranges over subsets of $\{1,\ldots,n\}$ and $\partial_J$ differentiates once in each coordinate indexed by $J$. Hence
\[
 \Big\|\sup_{Q_0}|\nabla Y|\Big\|_p\le C_Q\big(\sigma_L\sqrt p+\norm A_\op p\big),\qquad C_Q=3nk(1+k\ell_0)^n .
\]
With $\sigma_L^2\le\norm A_\op\le C'L^{-d}$, choosing $p\asymp L^d$ gives
\begin{equation}\label{eq:gradtail}
 \PP\Big\{\sup_{Q_0}|\nabla Y|\ge c_g\Big\}\le e^{-cL^d}.
\end{equation}

\emph{Conclusion.} Outside the event in~\eqref{eq:gradtail}, take $e\in\V_s\setminus\D_s$ and the nearest $v\in\D_s$. The segment $[e,v]$ lies in $\V_s$, and~\eqref{eq:gap} gives $\Rea S_L(e)-\Rea S_L(v)\ge(c_g-\sup|\nabla Y|)|e-v|>0$. So every minimizer is a deep hole.

\subsection{Gaussian comparison (Theorem~\ref{thm:localization}.3)}

Let $F_i=Y(v_i)/\sigma_L$ for the $N$ deep holes and $r=\norm A_\op/\sigma_L\le CL^{-d/2}$. Realify: $F_i=x^\top B_ix-\tr B_i$ with $x$ standard real Gaussian, $\norm{B_i}_\op\le r/2$, $\norm{B_i}_\HS\le1/\sqrt2$, and $C_{ij}=\E F_iF_j=2\tr(B_iB_j)$. By~\eqref{eq:nondeg2}, $C_{ii}\ge c$. Gaussian integration by parts gives $\E[F_if(F)]=\sum_j\E[T_{ij}\partial_jf(F)]$ with $T_{ij}=2x^\top\operatorname{sym}(B_iB_j)x$ and $\operatorname{sym}(M)=(M+M^\top)/2$, where $\E T_{ij}=C_{ij}$ and $\E|T_{ij}-C_{ij}|\le r$. With $G\sim N(0,C)$ independent and the interpolation $U_t=\sqrt tF+\sqrt{1-t}G$,
\[
 |\E h(F)-\E h(G)|\le\frac r2\sum_{i,j}\norm{\partial_{ij}h}_\infty .
\]
This allows singular $C$. Bracketing the indicator of $\{F_i>a_i\ \forall i\}$ by $C^2$ functions at scale $\delta$ costs $C\delta^{-2}r$ in the comparison. Each bracket differs from the indicator only on the union of the strips $\{|x_i-a_i|\le\delta\}$, so for the Gaussian vector the error is at most $\sum_i\PP\{|G_i-a_i|\le\delta\}\le2N\delta/\sqrt{2\pi c}$. Only the marginal variances $C_{ii}\ge c$ enter, and no lower bound on the smallest eigenvalue of $C$ is needed. Taking $\delta=r^{1/3}$, using $a_i=-\gamma_s/\sigma_L$, and combining with part~2 proves part~3.

\subsection{Limit law, sign classes, and sharp aperture}

\emph{Part 4.} Along $s_j\to s_*$, the normalized covariances converge to $G$ on $\D_*$ by~\eqref{eq:weak}. Their square roots converge as well, so coupled Gaussians converge almost surely. The smooth comparison above gives joint weak convergence. Each coordinate of the limit has positive variance, so $\{\min Z=-t\}$ is null, and the probabilities converge. For $t=\infty$, a union bound gives $N\Phi(-t_j)\to0$.

\emph{Proposition~\ref{prop:classes}.} Evenness of cosine gives $Z_v=Z_{-v}$. For distinct $x_1,\ldots,x_m\in\R^n$ with $n\ge2$, the functions $e^{\mathrm iku\cdot x_j}$ are linearly independent on the sphere. To see this, restrict to a great circle $p\cos z+q\sin z$ on which the values $q\cdot x_j$ are distinct, and continue analytically to $z=\mathrm i\tau$. As $\tau\to\infty$ the term with the smallest $q\cdot x_j$ dominates, which forces its coefficient to vanish. Hence the functions $\cos(ku\cdot v)$, one for each sign pair, are independent, and their Gram matrix is positive definite. $H$ is continuous because the marginals have densities. It is strictly decreasing because the Gaussian vector indexed by the sign classes has a positive density. For the square lattice the two classes give a bivariate normal with correlation $\rho=J_0(\sqrt2j_{0,1})/\frac12(1+J_0(2j_{0,1}))$, and the wedge formula gives~\eqref{eq:p0square}.

\emph{Theorem~\ref{thm:sharp}.} For fixed $s<s_*$, $\E\sup_{\V_s}|Y|^2\le C\sigma_L^2$, so the set of bad apertures is bounded. $\PP(F_L)$ is a right-continuous step function of $L$, so $L_\eta$ is attained. Let $0<\eta<p_0$, $q=q_\eta$ and $r_\gamma=(\kappa_nq^2/\gamma^2)^{1/d}$. At $L=(1-\epsilon)r_\gamma$, the uniform normalization gives $\gamma/\sigma_L\to q(1-\epsilon)^{d/2}<q$, so the limit probability exceeds $\eta$. Hence $L_\eta\ge(1-\epsilon)r_\gamma$ eventually. Conversely, suppose bad apertures $L_j\ge(1+\epsilon)r_{\gamma_j}$ existed along $\gamma_j\to0$. Then $\liminf\gamma_j/\sigma_{L_j}\ge q(1+\epsilon)^{d/2}$, and along a convergent subsequence the limit probability would be below $\eta$, a contradiction. The cases $\eta>p_0$ and $\eta=p_0$ follow by the same subsequence argument. For the square lattice, rational interval arithmetic on the Bessel series certifies $p_0<24/25$, so $\eta=0.99$ violates the restriction.

\subsection{Joint limit and relative error (Theorem~\ref{thm:joint})}

\emph{Exponential chaos.} For $Y=\sum_j\lambda_j(E_j-1)$ with i.i.d.\ unit exponentials, $v^2=\sum_j\lambda_j^2$ and $b=\max_j|\lambda_j|$, the cumulant generating function is $K(q)=\sum_j[-q\lambda_j-\log(1-q\lambda_j)]$ for $|q|b<1$, and $\kappa_r(Y)=(r-1)!\sum_j\lambda_j^r$. Hence $|\kappa_r(Y/v)|\le(r-1)!(b/v)^{r-2}\le(r-2)!(2b/v)^{r-2}$ for $r\ge3$, which is condition (S) of the cumulant method of D\"oring, Jansen and Schubert~\cite[Thm.~2.1]{djs} with $\Delta=v/(2b)$. The same cumulant bounds hold for $-Y$ and, after recentering and normalizing, under every exponential tilt with $1+\theta\lambda_j>0$, with $\lambda_j$ replaced by $\lambda_j/(1+\theta\lambda_j)$.

\emph{A logarithmic moderate-deviation lemma.} If $x/v\to\infty$ and $xb/v^2\to0$, then $\log\PP(Y\le-x)=-\frac{x^2}{2v^2}(1+o(1))$. The upper bound is Chernoff's inequality at $q=-x/v^2$, using $K(q)=\frac12v^2q^2(1+O(|q|b))$. For the lower bound let $y=x+3v$ and choose $\theta>0$ with $K'(-\theta)=-y$. Then $\theta=(y/v^2)(1+O(yb/v^2))$. Under the tilted law $\dd\mathbb Q/\dd\PP=\exp(-\theta Y-K(-\theta))$, $Y$ has mean $-y$ and variance at most $2v^2$, so $\mathbb Q\{-x-6v\le Y\le-x\}>\frac12$. Changing measure back gives $\PP(Y\le-x)\ge\frac12\exp\{K(-\theta)-\theta(x+6v)\}$, which is the claim.

\emph{Part 1.} At a deep hole, $\norm{H_v}_\op\le\norm A_\op\le CL^{-d}$ by Lemma~\ref{lem:op}, and $\nu_v^2=\sigma_L^2G(v,v)(1+o(1))=(\kappa_nG_*+o(1))L^{-d}$ uniformly by~\eqref{eq:weak}. So $\gamma_s\norm{H_v}_\op/\nu_v^2=O(\gamma_s)\to0$, and the lemma gives $\log\PP\{Y_L(v)\le-\gamma_s\}=-\gamma_s^2L^d/(2\kappa_nG_*)(1+o(1))$ for each deep hole. By part~2 of Theorem~\ref{thm:localization},
\[
 \max_v\PP\{Y_L(v)\le-\gamma_s\}\le\PP(F_L)\le\sum_v\PP\{Y_L(v)\le-\gamma_s\}+Ce^{-cL^d},
\]
and the last term is negligible on the scale $\gamma_s^2L^d=o(L^d)$. This proves the logarithmic asymptotics. For the aperture, put $w=\log(1/\eta)$ and $\ell_\pm^d=(1\pm\epsilon)2\kappa_nG_*\gamma_s^{-2}w$. At $L=\ell_-$ the asymptotics give $\PP(F_L)>\eta$ eventually, so $L_\eta\ge\ell_-$. For all $L\ge\ell_+$ simultaneously, the Bernstein bound $\PP\{Y_L(v)\le-\gamma_s\}\le\exp[-\gamma_s^2/(2(\nu_v^2+\norm{H_v}_\op\gamma_s))]$, which follows from $K(-\theta)\le\nu_v^2\theta^2/(2(1-\norm{H_v}_\op\theta))$, together with the uniform variance limit and the localization bound, gives $\PP(F_L)\le\eta$ eventually. Hence $L_\eta\le\ell_+$. Let $\epsilon\downarrow0$. No monotonicity of $\PP(F_L)$ in $L$ is used.

\emph{Part 2.} Throughout, $t=t_L$, $\sigma=\sigma_L$, and $\varepsilon=\norm{A}_\op/\sigma=O(L^{-d/2})$. Since the sample set is centrally symmetric, the sampling operator, written $\mathcal K=A$ in this paragraph (kernel $D_L(k(u-u'))$), is real symmetric. With $C_e,T_e$ the multiplications by $\cos(ku\cdot e)$ and $\sin(ku\cdot e)$,
\[
 H_{\pm e}=Q_e\pm A_e,\qquad Q_e=\tfrac12(C_e\mathcal K+\mathcal KC_e),\qquad A_e=\tfrac{\mathrm i}2(T_e\mathcal K-\mathcal KT_e),
\]
where $Q_e$ is real symmetric, $A_e$ is imaginary antisymmetric, and $\norm{Q_e}_\op,\norm{A_e}_\op\le\norm{A}_\op$. Write $q_e,a_e$ for the centered forms. Then $Y_L(\pm e)=q_e\pm a_e$, and the event for a sign pair is exactly $\{q_e-|a_e|\le-\gamma_s\}$.

\begin{lemma}[Sign difference]\label{lem:signdiff}
For $e$ in a compact set, uniformly in $s\in I$, $\tr A_e^2\le CL^{-3d/2}$. Hence $\tr A_e^2/\sigma^2\le CL^{-d/2}$.
\end{lemma}

\begin{proof}
Let $\zeta$ be an even mollifier supported in a ball of fixed radius $c_0$, and replace the weights $\mathbf 1_{B_L}(a)$ by $(\mathbf 1_{B_L}*\zeta)(a)$, keeping the normalization $M$. The difference $\mathcal R$ has coefficients of modulus at most one supported in the shell $L-c_0\le|a|\le L+c_0$. Split it into positive parts $\mathcal R_\pm$. The shell count of Appendix~\ref{app:normcov} gives $\tr\mathcal R_\pm\le CL^{-1}$, and Lemma~\ref{lem:op} gives $\norm{\mathcal R_\pm}_\op\le CL^{-d}$, so $\norm{\mathcal R}_\HS^2\le CL^{-d-1}\le CL^{-3d/2}$ for $n\in\{2,3\}$. Commutation with $T_e$ preserves this bound. For the smoothed kernel, Poisson summation and $|\widehat{\mathbf 1}_{B_L}(z)|\le CL^n(1+L|z|)^{-(n+1)/2}$ give
\[
 |\widetilde D_L(q)|^2\le C\sum_{G\in\G_s,\ |G|\le4k+1}(1+L|q-G|)^{-(n+1)}+CL^{-(n+1)}\qquad(|q|\le2k).
\]
The kernel of $A_e$ is $\frac{\mathrm i}2D_L(q)[\sin(ku\cdot e)-\sin(ku'\cdot e)]$ with $q=k(u-u')$. For $G\ne0$, the small-ball bounds for the law of $q$ (mass $Cr^{3/2}$ for $n=2$, $Cr^3$ for $n=3$, including tangency) and dyadic annuli give contributions $O(L^{-3/2})$ and $O(L^{-3})$, that is, $O(L^{-3d/2})$. For $G=0$, $|\sin(ku\cdot e)-\sin(ku'\cdot e)|^2\le C|q|^2$, and the radial density of $q$ near $0$ is at most $Cr^{d-1}$, so the contribution is at most $C\int_0^k r^{d+1}(1+Lr)^{-(d+2)}\dd r\le CL^{-d-2}\log(2+L)$.
\end{proof}

\begin{lemma}[Collapse of a sign pair]\label{lem:collapse}
If $t\to\infty$ and $t\,[L^{-d/4}\sqrt{\log L}+L^{-d/2}\log L]\to0$, then uniformly in $s\in I$
\[
 \PP\{Y_L(e)\le-\gamma_s\}=(1+o(1))\PP\{q_e\le-\gamma_s\}=(1+o(1))\PP\{\min(Y_L(e),Y_L(-e))\le-\gamma_s\},
\]
and the same holds for the centered Gaussian pair with the exact covariance of $(Y_L(e),Y_L(-e))$.
\end{lemma}

\begin{proof}
Let $\tau^2=\tr Q_e^2$ and $K(\theta)=\log\E e^{\theta q_e}$. Since $t\varepsilon\to0$, there is a unique small saddle $h=(\gamma_s/\tau^2)(1+O(t\varepsilon))$ with $K'(-h)=-\gamma_s$. Under $\dd\PP_h/\dd\PP=\exp\{-hq_e-K(-h)\}$, the complex Gaussian vector has covariance $R_h=(I+hQ_e)^{-1}$, which is real symmetric. Hence $\E_ha_e=\tr(R_hA_e)=0$, $\operatorname{Var}_ha_e\le4\tr A_e^2$, and the tilted operator norm of $a_e$ is at most $2\norm A_\op$. Put $U=q_e+\gamma_s$ and $I=h\gamma_s-K(-h)$. Under $\PP_h$, $U$ is a centered exponential chaos with cumulants controlled as above, so~\cite[Thm.~2.4]{djs} gives $\sup_x|\PP_h\{U/\tau_h\le x\}-\Phi(x)|\le C\varepsilon$ with $\tau_h^2=K''(-h)$. Writing $B_h=\E_h[e^{hU}\mathbf 1\{U\le0\}]$,
\[
 \PP\{q_e\le-\gamma_s\}=e^{-I}B_h,\qquad B_h=e^{h^2\tau_h^2/2}\Phi(-h\tau_h)+O(\varepsilon)\asymp t^{-1}.
\]
Let $a_L/\sigma=C_p[L^{-d/4}\sqrt{\log L}+L^{-d/2}\log L]$. The Bernstein bound of Lemma~\ref{lem:tail} under $\PP_h$ and Lemma~\ref{lem:signdiff} give $\PP_h\{|a_e|>a_L\}\le2L^{-p}$ for any fixed $p$, and $\E_he^{2h|a_e|}\le2\exp\{Ch^2\tr A_e^2\}=O(1)$ because $t^2L^{-d/2}\to0$. The pair event adds the region $\{0<U\le|a_e|\}$. On $\{|a_e|\le a_L\}$ its tilted weight is at most $e^{ha_L}\PP_h\{0<U\le a_L\}\le Ce^{ha_L}(a_L/\sigma+\varepsilon)$. On the complement, Cauchy--Schwarz bounds it by $CL^{-p/2}$. Dividing by $B_h\asymp t^{-1}$ gives a relative excess $O(t[e^{ha_L}(a_L/\sigma+\varepsilon)+L^{-p/2}])\to0$, since $ha_L=o(1)$. Take $p$ large. The same argument with the two-sided band $\{|U|\le|a_e|\}$ compares $\{Y_L(e)\le-\gamma_s\}$ with $\{q_e\le-\gamma_s\}$. For the Gaussian pair, the average and the half-difference are independent with variances $\tau^2$ and $\tr A_e^2$. The same proof applies with zero Berry--Esseen error. Finally $\operatorname{Var}Y_L(e)-\tau^2=O(\sigma^2L^{-d/2})$, and Mills' ratio shows that this variance change costs a factor $1+o(1)$ when $t^2L^{-d/2}\to0$.
\end{proof}

\begin{lemma}[Union over sign classes]\label{lem:union}
Under the assumption of Lemma~\ref{lem:collapse}, with one representative $v_i$ from each sign class,
\[
 \PP(F_L)=(1+o(1))\sum_i\PP\{q_{v_i}\le-\gamma_s\},
\]
and the same holds for the Gaussian proxy of Theorem~\ref{thm:localization}.3.
\end{lemma}

\begin{proof}
By Lemma~\ref{lem:signdiff} and~\eqref{eq:weak}, the averages $q_{v_i}$ have limiting normalized covariance equal to $G$ on one representative per class, which is positive definite with all diagonal entries $G_*$ (Proposition~\ref{prop:classes}). Fix $\bar\rho<1$ above the largest correlation between distinct classes. For large $L$, $\operatorname{Var}(q_{v_i}+q_{v_j})\le2G_*(1+\bar\rho)\sigma^2$, and Chernoff's bound for this exponential chaos gives
\[
 \PP\{q_{v_i}\le-\gamma_s,\ q_{v_j}\le-\gamma_s\}\le\PP\{q_{v_i}+q_{v_j}\le-2\gamma_s\}\le\exp\Big\{-\frac{(1+o(1))t^2}{G_*(1+\bar\rho)}\Big\}.
\]
Each single tail is $\exp\{-(1+o(1))t^2/(2G_*)\}$ with a prefactor polynomial in $t^{-1}$, so the pairwise intersections are of smaller order uniformly. Hence the union of the average events has probability $(1+o(1))\sum_i\PP\{q_{v_i}\le-\gamma_s\}$. The union of the exact pair events contains it, and its excess is at most the sum of the excesses bounded in Lemma~\ref{lem:collapse}. The localization error $Ce^{-cL^d}$ is negligible against $e^{-O(t^2)}$ in this range. The same steps apply to the Gaussian proxy.
\end{proof}

\emph{Conclusion of Part 2.} The condition of Lemma~\ref{lem:collapse} is equivalent to $t=o(L^{d/4}/\sqrt{\log L})$. For each scalar form, the cumulant theorem~\cite[Thm.~2.1]{djs} gives
\[
 \frac{\PP\{Y\le-x\}}{\Phi(-z)}=\exp\Big\{-\frac{c_3z^3}{6}+O(\delta^2z^4)\Big\}\big[1+O((1+z)\delta)\big],\qquad z=\frac xv,\ c_3=\frac{2\sum_j\lambda_j^3}{v^3},\ \delta=\frac bv,
\]
for $0\le z\le c/\delta$. With $\delta=O(L^{-d/2})$ and $z\asymp t$, one has $\delta^2z^4=O(L^{-d}t^4)\to0$ in this range, so Lemmas~\ref{lem:collapse}--\ref{lem:union} give the cubic-corrected formula, with exact finite-$L$ variances $\nu_v^2$ and third cumulants $\kappa_{3,v}=2\tr H_v^3$. Replacing $Q_e$ by $H_e$ in the cubic term changes it by $O(t^3\varepsilon L^{-d/2})=o(1)$, since $\tr H_e^3-\tr Q_e^3=3\tr(Q_eA_e^2)$ and $|\tr(Q_eA_e^2)|\le\norm A_\op\tr A_e^2$. When $t=o(L^{d/6})$, the cubic term is itself $o(1)$, and comparing with the same expansion for the Gaussian proxy gives the ratio statement. In the Gaussian limit the same union argument gives $H(q)=N_*\Phi(-q/\sqrt{G_*})(1+o(1))$ as $q\to\infty$, and inverting $\Phi(-x)=\Phi'(x)x^{-1}(1+O(x^{-2}))$ gives the quantile expansion.

\subsection{No uniform rate over spectral densities}

Let $n=2$, $K=B_1$, $\Lambda=\frac14\Z^2$, and $f_m(\theta)=1+\frac12\cos(m\theta)$ with $m$ even. Then $\phi_m(r)=J_0(2\pi|r|)+\frac12\mathrm i^mJ_m(2\pi|r|)\cos(m\arg r)$. The continuum constant is $\kappa(f_m)=\frac{8}{3\pi^3}(1+\frac18)$ for every $m$. For $|t|\le4\pi L$, $|J_m(t)|\le\delta(L,m)=\frac{(2\pi L)^m}{m!}\exp\big(\frac{4\pi^2L^2}{m+1}\big)$. With $m_L=2\lceil50(L+1)^2\rceil$, $L\delta(L,m_L)\to0$, so $L\sigma_L^2(f_{m_L})-L\sigma_L^2(1)\to0$, while the limits differ by $\frac{1}{3\pi^3}$. Hence $\liminf_L\sup_f|L\sigma_L^2(f)-\kappa(f)|\ge\frac1{6\pi^3}$ over this family. Each $f_m$ is smooth and satisfies $\frac12\le f_m\le\frac32$, so bounds on $f$ alone give no rate.

\section{Proofs for Section~\ref{sec:specular}}\label{app:specular}

\subsection{Uniform limit}
Expand $S_L=(1-\alpha)S_c+S_\times+S_{\rm at}$. Here $S_c$ is built from $h_c$, $S_\times$ collects the cross terms
\[
 \sqrt{(1-\alpha)p_j}\,\frac1M\sum_{a\in\A}h_c(a)\,\bar g_j\,e^{-\mathrm i\xi_j\cdot(a+e)}
\]
and their counterparts with the roles of $h_c$ and the paths exchanged, and
\[
 S_{\rm at}(e)=\sum_{j,l}\sqrt{p_jp_l}\,g_j\bar g_l\,e^{-\mathrm i\xi_l\cdot e}\,D_L(\xi_j-\xi_l).
\]
By~\eqref{eq:uniformtail} and Lemma~\ref{lem:op}, $\sup_\V|S_c-\phi_c(-\cdot)|=O_P(L^{-d/2})$. Conditionally on the gains, each cross term is a linear Gaussian functional of $h_c$. Its variance is $\int|D_L(ku-\xi_j)|^2\mu_c(\dd u)\le C\norm{A}_\op=O(L^{-d})$, and the same holds for its $e$-derivatives. A Sobolev bound on the compact cell, applied to real and imaginary parts, then gives $\sup_\V|S_\times|=O_P(L^{-d/2})$. For $q\in\G$, $D_L(q)=1$. For $q\notin\G$, choose a lattice vector $a_0$ with $e^{\mathrm ia_0\cdot q}\ne1$ and write $D_L(q)(1-e^{\mathrm ia_0\cdot q})$ as a sum over the symmetric difference of $\A$ and $\A+a_0$. This gives $|D_L(q)|=O(L^{-1})$, with a constant depending on the fixed frequency gaps. Hence
\begin{equation}\label{eq:Sinf}
 S_\infty(e)=(1-\alpha)\phi_c(-e)+\sum_{\mathcal C}\Big(\sum_{j\in\mathcal C}\sqrt{p_j}g_j\Big)\overline{\Big(\sum_{l\in\mathcal C}\sqrt{p_l}g_le^{\mathrm i\xi_l\cdot e}\Big)},
\end{equation}
whose real part is $m_c(e)+g^*B_eg$ with $B_e$ as defined in Section~\ref{sec:specular}.

\subsection{Theorem~\ref{thm:specular}}
Let $M_\infty=\min_\V(m_c+g^*B_eg)$ and $g=\sqrt TU$. Then $M_\infty\le0$ if and only if $T\,(-U^*B_eU)\ge m_c(e)$ for some $e$, that is, $T\ge1/r(U)$ when $r(U)>0$. Since $T$ has the Gamma density $t^{J-1}e^{-t}/(J-1)!$ and is independent of $U$,
\[
 \PP(M_\infty\le0\mid U)=\PP\big(T\ge1/r(U)\big)=H_J\big(r(U)\big),\qquad \PP(M_\infty=0\mid U)=0 .
\]
The boundary event is null, so uniform convergence in probability gives $\PP(F_L)\to\PP(M_\infty\le0)=\E H_J(r(U))$. If some $B_{e_0}$ has a negative eigenvalue, then $r(U)>0$ for $U$ in a neighborhood of the corresponding eigenvector, a set of positive probability, so the limit is positive. If every $B_e$ is positive semidefinite, then $r\equiv0$. Without aliasing $B_e$ is diagonal, which gives Corollary~\ref{cor:nothreshold}. For one path, $J=1$, $r=p/((1-p)\kappa)$ and $H_1(r)=e^{-1/r}$.

\subsection{Proposition~\ref{prop:fixed} and Corollary~\ref{cor:quarter}}
In model (D) without aliasing, $|g_j|^2=1$ and the cross terms between different paths vanish in the limit, so $S_\infty$ is deterministic. Uniform convergence gives the dichotomy. For one path, $\min_\V[(1-p)\Rea\phi_c+p\cos(\xi\cdot e)]>0$ if and only if $(1-p)\kappa>p$.

For Corollary~\ref{cor:quarter}, $|u\cdot e|\le|e|\le1/4$ gives $\cos(2\pi u\cdot e)\ge0$. Conversely, if $R>1/4$, a point on the segment from $0$ to a deep hole with radius in $(1/4,\min\{R,3/4\})$ has negative cosine along its own direction. With the dual lattice $\Lambda^\#=B^{-T}\Z^n$, so that $\G=2\pi\Lambda^\#$, an alias $u_j-u_l=q\in\Lambda^\#\setminus\{0\}$ needs $|q|\le2$. But every nonzero dual vector satisfies $R>1/(2|q|)$: take the point $q/(2|q|^2)$, shifted transversally off the countable set of lattice projections. All lattice points are at longitudinal distance at least $1/(2|q|)$ from it. Hence $R\le1/4$ excludes aliasing. The densities follow from the covering theorems~\cite{bambah,barnes56}.

\subsection{Theorem~\ref{thm:known}}
\emph{Upper bound.} $C$ is a centrally symmetric convex body. If $\V(\Lambda)\subset C$, the translates of $C$ cover the plane. Under the covering condition $\Lambda+C=\R^2$, this holds by assumption. By F\'ary's theorem~\cite[\S III.1, Lemme~1]{fary}, $C$ then contains a centrally symmetric hexagon, possibly degenerate, that tiles the plane under $\Lambda$, so $\covol\Lambda$ is at most the largest area of such a hexagon in $C$.

\emph{Largest inscribed hexagon.} Center the hexagon. For vectors $p,q,r$,
\[
 |\conv\{\pm p,\pm q,\pm r\}|=\max\{2D_{pq},2D_{qr},2D_{rp},D_{pq}+D_{qr}+D_{rp}\},\qquad D_{pq}=|\det(p,q)|,
\]
which is convex in each vector. So the vertices can be moved to the circular arcs of $\partial C$. Let $\theta_1\le\theta_2\le\theta_3$ be the angles of the three vertices in the upper half-plane, so that $\theta_0\le\theta_1$ and $\theta_3\le\pi-\theta_0$ with $\theta_0=\arccos(h/\rho)$, and put $a=\theta_2-\theta_1$, $b=\theta_3-\theta_2$. The area is $\rho^2[\sin a+\sin b+\sin(a+b)]\le\rho^2[2\sin(w/2)+\sin w]$ with $w=a+b\le2\arcsin(h/\rho)\le2\pi/3$. The bound is increasing in $w$, so the area is at most $2h(\rho+s)$. Equality holds for the vertices $(0,\pm\rho)$ and $(\pm h,\pm s)$.

\emph{The Voronoi cell of $\Lambda_*$.} With $a=(2h,0)$, $b=(h,t)$ and $t=\rho+s$, the identities $h^2+t^2=2\rho t$ and $h^2/t=\rho-s$ show that the bisectors of $\pm a$, $\pm b$, $\pm(b-a)$ cut out exactly this hexagon. For a general lattice vector $(hA,tn)$, with $A$ and $n$ of equal parity, one checks the support function at the six vertices: $\max\{\rho t|n|,\,h^2|A|+st|n|\}\le\frac12(h^2A^2+t^2n^2)$, using $t\ge h$. So the hexagon is the Voronoi cell, and its area is $2ht$.

\emph{Strict margin and boundary.} A uniform positive continuous margin confines residuals to $C_{R',h}=\{e:|e|\le R',\ |e_1|\le h\}$ with $R'<\rho$. The same argument then gives an area strictly below $2h(\rho+s)$, so the infimum is not attained. The Voronoi cell of $(1-\varepsilon)\Lambda_*$ is $(1-\varepsilon)\V(\Lambda_*)$. On it $J_0(2\pi|e|)\ge J_0(2\pi(1-\varepsilon)\rho)>0$ and $\cos(2\pi e_1)\ge\sin(\pi\varepsilon/2)>0$, so Theorem~\ref{thm:specular} gives $\PP(F_L)\to0$. For $\Lambda_*$ itself, at the vertex $(h,s)$ both $J_0(2\pi\rho)$ and $\cos(2\pi h)$ vanish, so $\E\Rea S_L=0$ there for every $L$. In model (G) the field is Gaussian, so at that vertex $X=\Rea S_L=\sum_j\lambda_j(E_j-1)$ as in Lemma~\ref{lem:tail}. The continuous part gives $s_X^2=\E X^2>0$, and $\E X^4\le9s_X^4$. Then $\E|X|\ge s_X/3$, $\E X_-\ge s_X/6$, and $\E X_-\le\sqrt{\E X_-^2\,\PP(X<0)}\le s_X\sqrt{\PP(X<0)}$, so $\PP(F_L)\ge\PP(X<0)\ge1/36$ for every finite nonempty aperture.

\bibliographystyle{IEEEtran}
\bibliography{refs}

\begin{thebibliography}{10}
\providecommand{\url}[1]{#1}
\csname url@samestyle\endcsname
\providecommand{\newblock}{\relax}
\providecommand{\bibinfo}[2]{#2}
\providecommand{\BIBentrySTDinterwordspacing}{\spaceskip=0pt\relax}
\providecommand{\BIBentryALTinterwordstretchfactor}{4}
\providecommand{\BIBentryALTinterwordspacing}{\spaceskip=\fontdimen2\font plus
\BIBentryALTinterwordstretchfactor\fontdimen3\font minus
  \fontdimen4\font\relax}
\providecommand{\BIBforeignlanguage}[2]{{%
\expandafter\ifx\csname l@#1\endcsname\relax
\typeout{** WARNING: IEEEtran.bst: No hyphenation pattern has been}%
\typeout{** loaded for the language `#1'. Using the pattern for}%
\typeout{** the default language instead.}%
\else
\language=\csname l@#1\endcsname
\fi
#2}}
\providecommand{\BIBdecl}{\relax}
\BIBdecl

\bibitem{ota}
Y.~Yang, Z.~Zhang, Y.~Tian, Z.~Yang, C.~Huang, C.~Zhong, and K.-K. Wong,
  ``Over-the-air split machine learning in wireless {MIMO} networks,''
  \emph{IEEE Journal on Selected Areas in Communications}, vol.~41, no.~4, pp.
  1007--1022, 2023, doi:
  \href{https://doi.org/10.1109/JSAC.2023.3242701}{10.1109/JSAC.2023.3242701}.

\bibitem{wong}
K.-K. Wong, A.~Shojaeifard, K.-F. Tong, and Y.~Zhang, ``Fluid antenna
  systems,'' \emph{IEEE Transactions on Wireless Communications}, vol.~20,
  no.~3, pp. 1950--1962, 2021.

\bibitem{bambah}
R.~P. Bambah, ``On lattice coverings by spheres,'' \emph{Proceedings of the
  National Institute of Sciences of India}, vol.~20, pp. 25--52, 1954.

\bibitem{barnes56}
E.~S. Barnes, ``The covering of space by spheres,'' \emph{Canadian Journal of
  Mathematics}, vol.~8, pp. 293--304, 1956, doi:
  \href{https://doi.org/10.4153/CJM-1956-033-4}{10.4153/CJM-1956-033-4}.

\bibitem{conway}
J.~H. Conway and N.~J.~A. Sloane, ``Voronoi regions of lattices, second moments
  of polytopes, and quantization,'' \emph{IEEE Transactions on Information
  Theory}, vol.~28, no.~2, pp. 211--226, 1982, doi:
  \href{https://doi.org/10.1109/TIT.1982.1056483}{10.1109/TIT.1982.1056483}.

\bibitem{barnes}
E.~S. Barnes and N.~J.~A. Sloane, ``The optimal lattice quantizer in three
  dimensions,'' \emph{SIAM Journal on Algebraic and Discrete Methods}, vol.~4,
  no.~1, pp. 30--41, 1983, doi:
  \href{https://doi.org/10.1137/0604005}{10.1137/0604005}.

\bibitem{pm62}
D.~P. Petersen and D.~Middleton, ``Sampling and reconstruction of
  wave-number-limited functions in {$N$}-dimensional {E}uclidean spaces,''
  \emph{Information and Control}, vol.~5, no.~4, pp. 279--323, 1962, doi:
  \href{https://doi.org/10.1016/S0019-9958(62)90633-2}{10.1016/S0019-9958(62)90633-2}.

\bibitem{kunsch}
H.~R. K{\"u}nsch, E.~Agrell, and F.~A. Hamprecht, ``Optimal lattices for
  sampling,'' \emph{IEEE Transactions on Information Theory}, vol.~51, no.~2,
  pp. 634--647, 2005, doi:
  \href{https://doi.org/10.1109/TIT.2004.840864}{10.1109/TIT.2004.840864}.

\bibitem{agrell18}
E.~Agrell and B.~Cs{\'e}bfalvi, ``Multidimensional sampling of isotropically
  bandlimited signals,'' \emph{IEEE Signal Processing Letters}, vol.~25, no.~3,
  pp. 383--387, 2018, doi:
  \href{https://doi.org/10.1109/LSP.2017.2720143}{10.1109/LSP.2017.2720143}.

\bibitem{fary}
I.~F{\'a}ry, ``Sur la densit{\'e} des r{\'e}seaux de domaines convexes,''
  \emph{Bulletin de la Soci{\'e}t{\'e} Math{\'e}matique de France}, vol.~78,
  pp. 152--161, 1950, doi:
  \href{https://doi.org/10.24033/bsmf.1413}{10.24033/bsmf.1413}.

\bibitem{adlertaylor}
R.~J. Adler and J.~E. Taylor, \emph{Random Fields and Geometry}.\hskip 1em plus
  0.5em minus 0.4em\relax New York: Springer, 2007.

\bibitem{azais}
J.-M. Aza{\"\i}s and M.~Wschebor, \emph{Level Sets and Extrema of Random
  Processes and Fields}.\hskip 1em plus 0.5em minus 0.4em\relax Hoboken, NJ:
  Wiley, 2009.

\bibitem{marinucci}
D.~Marinucci and G.~Peccati, \emph{Random Fields on the Sphere: Representation,
  Limit Theorems and Cosmological Applications}.\hskip 1em plus 0.5em minus
  0.4em\relax Cambridge University Press, 2011.

\bibitem{nourdinpeccati}
I.~Nourdin and G.~Peccati, \emph{Normal Approximations with {M}alliavin
  Calculus: From {S}tein's Method to Universality}.\hskip 1em plus 0.5em minus
  0.4em\relax Cambridge University Press, 2012.

\bibitem{kkw}
M.~Krishnapur, P.~Kurlberg, and I.~Wigman, ``Nodal length fluctuations for
  arithmetic random waves,'' \emph{Annals of Mathematics}, vol. 177, no.~2, pp.
  699--737, 2013, doi:
  \href{https://doi.org/10.4007/annals.2013.177.2.8}{10.4007/annals.2013.177.2.8}.

\bibitem{berry}
M.~V. Berry, ``Regular and irregular semiclassical wavefunctions,''
  \emph{Journal of Physics A: Mathematical and General}, vol.~10, no.~12, pp.
  2083--2091, 1977.

\bibitem{gps}
L.~Gass, G.~Peccati, and M.~Stecconi, ``Scars in random waves and the
  {FGF}$(1/2)$ universality class,'' arXiv:2606.07842, 2026.

\bibitem{beliaev}
D.~Beliaev and A.~Hegde, ``Excursion fluctuations and spectral universality in
  {G}aussian fields,'' arXiv:2606.15630, 2026.

\bibitem{maini}
L.~Maini, ``Asymptotic covariances for functionals of weakly stationary random
  fields,'' \emph{Stochastic Processes and their Applications}, vol. 170, 2024,
  art. no. 104297, doi:
  \href{https://doi.org/10.1016/j.spa.2024.104297}{10.1016/j.spa.2024.104297}.

\bibitem{maininourdin}
L.~Maini and I.~Nourdin, ``Spectral central limit theorem for additive
  functionals of isotropic and stationary {G}aussian fields,'' \emph{The Annals
  of Probability}, vol.~52, no.~2, pp. 737--763, 2024, doi:
  \href{https://doi.org/10.1214/23-AOP1669}{10.1214/23-AOP1669}.

\bibitem{ferris}
M.~C. Ferris, ``Weak sharp minima and exact penalty functions,'' Computer
  Sciences Department, University of Wisconsin--Madison, Tech. Rep. 779, 1988.

\bibitem{carcamo}
J.~C{\'a}rcamo, A.~Cuevas, and L.-A. Rodr{\'\i}guez, ``Directional
  differentiability for supremum-type functionals: statistical applications,''
  \emph{Bernoulli}, vol.~26, no.~3, pp. 2143--2175, 2020, doi:
  \href{https://doi.org/10.3150/19-BEJ1188}{10.3150/19-BEJ1188}.

\bibitem{ngo14}
H.~Q. Ngo, E.~G. Larsson, and T.~L. Marzetta, ``Aspects of favorable
  propagation in massive {MIMO},'' in \emph{Proc. European Signal Processing
  Conference (EUSIPCO)}, 2014, pp. 76--80.

\bibitem{liray}
S.~Li, P.~J. Smith, P.~A. Dmochowski, H.~Tataria, M.~Matthaiou, and J.~Yin,
  ``Massive {MIMO} asymptotics for ray-based propagation channels,'' \emph{IEEE
  Transactions on Wireless Communications}, vol.~19, no.~6, pp. 3977--3991,
  2020, doi:
  \href{https://doi.org/10.1109/TWC.2020.2979146}{10.1109/TWC.2020.2979146}.

\bibitem{truong}
K.~T. Truong and R.~W. Heath, Jr., ``Effects of channel aging in massive {MIMO}
  systems,'' \emph{Journal of Communications and Networks}, vol.~15, no.~4, pp.
  338--351, 2013, doi:
  \href{https://doi.org/10.1109/JCN.2013.000065}{10.1109/JCN.2013.000065}.

\bibitem{clarke}
R.~H. Clarke, ``A statistical theory of mobile-radio reception,'' \emph{Bell
  System Technical Journal}, vol.~47, no.~6, pp. 957--1000, 1968.

\bibitem{rv13}
M.~Rudelson and R.~Vershynin, ``{Hanson--Wright} inequality and sub-gaussian
  concentration,'' \emph{Electronic Communications in Probability}, vol.~18,
  pp. 1--9, 2013, doi:
  \href{https://doi.org/10.1214/ECP.v18-2865}{10.1214/ECP.v18-2865}.

\bibitem{zenodo}
T.~Guo, H.~Luo, X.~Li, and Y.~Yang, ``Deep holes and aperture: certified
  whole-cell minima and code for the numerical comparison ({T}able {II}),''
  Zenodo, version 1.0, 2026, doi:
  \href{https://doi.org/10.5281/zenodo.23687425}{10.5281/zenodo.23687425}.

\bibitem{stein}
E.~M. Stein and R.~Shakarchi, \emph{Fourier Analysis: An Introduction}.\hskip
  1em plus 0.5em minus 0.4em\relax Princeton University Press, 2003.

\bibitem{durrett}
R.~Durrett, \emph{Probability: Theory and Examples}, 5th~ed.\hskip 1em plus
  0.5em minus 0.4em\relax Cambridge University Press, 2019.

\bibitem{djs}
H.~D{\"o}ring, S.~Jansen, and K.~Schubert, ``The method of cumulants for the
  normal approximation,'' \emph{Probability Surveys}, vol.~19, pp. 185--270,
  2022.

\end{thebibliography}
\end{document}